\documentclass[10pt,twocolumn,aps,pra,superscriptaddress]{revtex4-1}

\usepackage{graphicx}
\usepackage{color}
\usepackage{xcolor}
\usepackage{amssymb}
\usepackage{amsmath}
\usepackage{bm}
\usepackage[american]{babel}
\usepackage{hyperref}
\usepackage{braket}
\usepackage[T1]{fontenc} 
\usepackage{txfonts}
\usepackage{comment}
\usepackage{textcomp}
\usepackage[ruled,vlined]{algorithm2e}

\DeclareGraphicsExtensions{.pdf,.png,.jpg}
\hypersetup{colorlinks=true,urlcolor=teal,breaklinks=true,citecolor=teal,linkcolor=teal,pdfstartview=FitH,pdfpagemode=UseNone}

\newtheorem{lemma}{Lemma} 
\newcommand{\qed}{\hfill $\square$} 

\begin{document}

\title{Optimal Decoder for the Error Correcting Parity Code}

\author{Konstantin Tiurev}
\affiliation{Parity Quantum Computing Germany GmbH, 20095 Hamburg, Germany}
\author{Christophe Goeller}
\affiliation{Parity Quantum Computing Germany GmbH, 20095 Hamburg, Germany}
\affiliation{Parity Quantum Computing France SAS, 75016 Paris, France}
\author{Leo Stenzel}
\affiliation{Parity Quantum Computing Germany GmbH, 20095 Hamburg, Germany}
\author{Paul Schnabl}
\affiliation{Institute for Theoretical Physics, University of Innsbruck, A-6020 Innsbruck, Austria}
\author{Anette Messinger}
\affiliation{Parity Quantum Computing GmbH, A-6020 Innsbruck, Austria}
\author{Michael Fellner}
\affiliation{Parity Quantum Computing GmbH, A-6020 Innsbruck, Austria}
\author{Wolfgang Lechner}
\affiliation{Parity Quantum Computing Germany GmbH, 20095 Hamburg, Germany}
\affiliation{Parity Quantum Computing France SAS, 75016 Paris, France}
\affiliation{Institute for Theoretical Physics, University of Innsbruck, A-6020 Innsbruck, Austria}
\affiliation{Parity Quantum Computing GmbH, A-6020 Innsbruck, Austria}

\date{\today}

\begin{abstract}
We present a two-step decoder for the parity code and evaluate its performance in code-capacity and faulty-measurement settings. For noiseless measurements, we find that the decoding problem can be reduced to a series of repetition codes while yielding near-optimal decoding for intermediate code sizes and achieving optimality in the limit of large codes. 
In the regime of unreliable measurements, the decoder demonstrates fault-tolerant thresholds above 5\% at the cost of decoding a series of independent repetition codes in ${(1+1)}$ dimensions. Such high thresholds, in conjunction with a practical decoder, efficient long-range logical gates, and suitability for planar implementation, position the parity architecture as a promising candidate for demonstrating quantum advantage on qubit platforms with strong noise bias.
\end{abstract}

\maketitle

\section{Introduction} \label{sec:introduction}

Quantum computers have the potential to solve certain types of computational problems with polynomial or even super-polynomial speedups compared to their classical counterpart~\cite{9781107002173, dalzell2023quantumalgorithmssurveyapplications,doi:10.1137/S0097539795293172}. However, the inherent fragility of quantum information makes quantum algorithms highly susceptible to noise. To counter this, quantum error correction~(QEC)~\cite{RevModPhys.87.307, QEC2,Devitt_2013, doi:10.1080/00107514.2019.1667078} serves as an indispensable building block for enabling the construction of scalable and fault-tolerant quantum computers in the medium to long term. 

In current physical implementations of quantum devices, typically only direct interactions between nearest-neighbor qubits are permitted due to the planar layout. Any gate between distant qubits requires additional operations used to bring these qubits in direct contact
~\cite{PRXQuantum.4.010313,7059001,childs_et_al:LIPIcs.TQC.2019.3,Saeedi2011}. The limited connectivity between qubits imposes stringent constraints on the design of quantum algorithms and error correcting codes~\cite{Kitaev-AnnPhys-2003,TopologicalQuantumMemory,PhysRevLett.108.180501,PhysRevA.86.032324,PhysRevLett.97.180501,LandahlColor,Kubica2018,Bombin_2015}. The parity code we consider in this work offers an alternative solution for implementing long-range logical gates by allowing for any desired logical connectivity from nearest-neighbor physical interactions only~\cite{PhysRevLett.129.180503,doi:10.1126/sciadv.1500838,klaver2024swaplessimplementationquantumalgorithms, messinger2023}. Being an instance of low-density parity codes~(LDPCs)~\cite{panteleev2022asymptoticallygoodquantumlocally}, the parity code provides intrinsic protection against either phase-flip or bit-flip errors. Hence, both error correction and long-range logical gates can be realized via local operations only. When constructed upon qubits with strong noise bias characteristics coupled to bias-preserving gates, such as cat qubits~\cite{PhysRevX.9.041053,PRXQuantum.2.030345,PRXQuantum.2.030345,doi:10.1126/sciadv.aay5901,PhysRevResearch.4.013082}, or concatenated with, e.g., a repetition code, the parity architecture offers low-depth fault-tolerant implementations of certain quantum algorithms~\cite{messinger2024faulttolerantquantumcomputingparity}. 

Fault-tolerant architectures must be designed in conjunction with classical co-processors that decode quantum error correction measurement information in real-time~\cite{RevModPhys.87.307}. The complexity of decoding algorithms is crucial, as the quantum state is vulnerable to accumulating additional errors during the decoder operation as well as to avoid the backlog problem. Typically, improving decoding performance comes at the cost of increased computational complexity, creating a trade-off between accuracy and runtime. To maintain the fragile quantum information, decoders must be both highly accurate and sufficiently fast to correct errors before further accumulation of errors or complete decoherence occurs. The optimal decoder aims to infer the most probable occurred error from the syndrome with maximal success probability. However, achieving this for large codes is not viable since the complexity of decoding a quantum code is typically a $\#P$-complete problem~\cite{9456887,7097029}, and various approximate methods have to be developed to address this scalability challenge. The list includes, but is not limited to, minimum weight perfect matching~(MWPM)~\cite{10.1063/1.1499754}, union-find~\cite{Delfosse_2021,PhysRevA.102.012419,wu2022interpretationunionfinddecoderweighted}, maximum-likelihood~\cite{PhysRevA.90.032326,PhysRevX.9.041031}, belief-propagation~\cite{wolanski2025ambiguityclusteringaccurateefficient}, and neural-network~\cite{PhysRevLett.119.030501,PhysRevLett.128.080505} decoders. Previously, the error correcting capabilities of the parity architecture have only been studied in a quantum annealing setting under ideal measurements and a readout decoder~\cite{PhysRevA.93.052325, nambu2024errorcorrectionencodedquantum,nambu2024errorcorrectionparityencodingbasedannealing}. However, a universal parity-based quantum processor of Ref.~\cite{PhysRevLett.129.180503} requires non-destructive stabilizer measurements and a decoder capable of resolving the observed syndrome.

In this paper, we thoroughly investigate symmetrical and statistical properties of the parity code in the presence of noise. Using these properties, we construct a two-step parity decoder of the noise syndrome. We start by identifying one-dimensional materialized symmetries of the code, leading to the syndrome parity conservation law. A correction operator compatible with the syndrome is then constructed by pairing defects along the respective code symmetries, which constitutes the first step of the parity decoder. The second step of the decoder aims to find the most probable correction among all candidates by direct minimization. Strikingly, we show that both steps of the decoder can be reduced to a series of repetition codes, while yielding near-optimal decoding capabilities for intermediate-size codes and achieving optimality in the limit of large code distances $d$. Similarly, in the noisy measurements setting, the decoder exhibits the same computational complexity as decoding $d$ independent repetition codes in $(1+1)$ dimensions and achieves a fault-tolerant threshold which monotonically increases with the code size and saturates above 5\%. 

The remainder of this paper is organized as follows. In Sec.~\ref{sec:parity-code}, we outline the parity code for universal quantum computing. In
Sec.~\ref{sec:symmetries}, we identify the materialized symmetries of the code, which allows us to construct matching graphs for pairing the code defects. In Sec.~\ref{sec:statistical-properties}, we investigate the statistical properties of the logical operators of the code and prove that the parity code has an optimal threshold of 50\% under ideal measurements in the limit of large codes. In Sec.~\ref{sec:1-line-decoder}, we show that optimal decoding can be achieved within runtime scaling linearly with the code size when large codes are considered. Using the properties of the parity code investigated in previous sections we construct the two-step parity decoder in Sec.~\ref{sec:decoder} and numerically simulate its performance in Sec.~\ref{sec:code-capacity}. In Sec.~\ref{sec:faulty-measurements}, the decoder is adapted for and benchmarked in the regime of noisy measurements. Complexity, runtime, and parallelization of the decoder are discussed in Sec.~\ref{sec:complexity}. We review the results and conclude with future perspectives in Sec.~\ref{sec:outlook}.

\section{The parity code} \label{sec:parity-code}

The parity code~\cite{PhysRevLett.129.180503} is a LDPC code where the stabilizers are chosen such that certain connectivity between logical qubits is obtained. We focus on the parity code in the modified Lechner--Hauke--Zoller~(LHZ) configuration~\cite{PhysRevLett.129.180503,doi:10.1126/sciadv.1500838} illustrated in Fig.~\ref{fig:parity-code}. The code is formed by physical qubits placed at the vertices of a square lattice with triangular boundaries. The codespace is stabilized by weight-4~(weight-3 along the bottom boundary) operators
$G_i = \prod_{v \in \partial f_i} Z_v$
associated with the faces, or plaquettes, of the lattice. Here, $\partial f_i$ denotes all qubits in the support of plaquette $f_i$ and $Z_v$ denotes a Pauli-$Z$ operator acting on qubit $v$. A computational space is spanned by the simultaneous +1 eigenstate of all stabilizers, $
    \ket{\Psi}
    =
    (+1) G_i \ket{\Psi}$. 
A single Pauli-$X$ error on a data qubit anticommutes with all stabilizers $G_i$ it has support on, that is, $G_i X_v = (-1)X_v G_i$ if ${v \in \partial f_i}$. We will refer to individual stabilizers flipped by an error as defects. Collectively, the measured defects will be referred to as syndrome. The task of a decoder is to deduce the most probable configuration of errors compatible with the observed syndrome. As all stabilizers are Pauli-$Z$ products, the parity code alone corrects only bit-flip errors and has to be either constructed upon physical qubits with strong intrinsic noise bias~\cite{messinger2024faulttolerantquantumcomputingparity} or concatenated with another error correcting code. As introduced in 
Ref.~\cite{messinger2024faulttolerantquantumcomputingparity}, the code admits fault-tolerant implementations of a full universal gate set for logical qubits in the presence of infinite-biased noise. In the remainder of this article, we consider the regime of bit flips only. 

\begin{figure}[t]
\includegraphics[width=0.95\columnwidth]{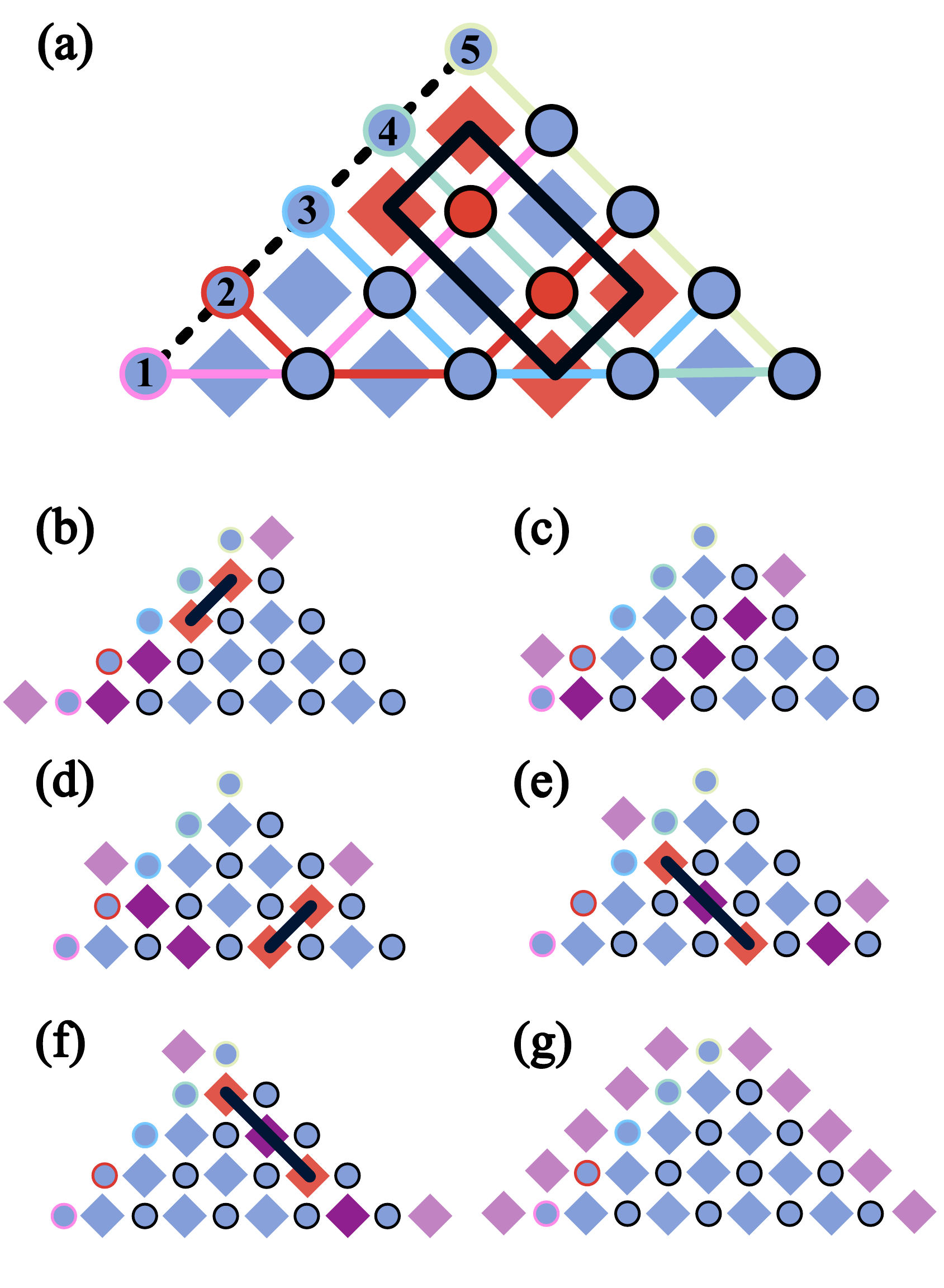}
\caption{\label{fig:parity-code} (a)~Distance-5 LHZ code. Qubits~(circles) lie on the vertices of the square lattice. Plaquettes correspond to stabilizers measuring the parity of the surrounding qubits in the $Z$ basis. The code contains five logical qubits. Logical Pauli-$Z$ operators correspond to single-qubit Pauli-$Z$ operators applied to base qubits labelled 1 to 5. Base qubits are arranged along the dashed line, which we refer to as qubit line 0. Logical Pauli-$X$ operators are formed by products of $X$ operators applied to ${d=5}$ qubits along solid lines of the same color, which we refer to as qubit lines 1 to $d$. A syndrome~(red plaquettes) triggered by qubit errors~(red circles) forms a closed contour~(solid black line) that encloses the flipped qubits. (b)--(f)~Symmetries of the distance-5 LHZ code formed by stabilizers shown in purple and red. Each qubit is in the support of either two stabilizers of a given symmetry or none. Stabilizers triggered along each symmetry by the error of panel~(a) are shown in red and connected by the shortest-path matching~(black lines). Virtual stabilizers shown with pale purple plaquettes are appended to ends of each symmetry line in (b)--(f) and form a virtual symmetry~(g). Combined together, pairings of panels (b)--(g) form a cluster of panel (a) enclosing the error. In panels (b)--(g), qubit lines are omitted to avoid cluttering.}
\end{figure}

In the LHZ layout, a distance-$d$ parity code encodes $d$ logical qubits in ${n = d(d+1)/2}$ physical data qubits using $n-d$ independent stabilizers. Logical Pauli-$Z$ operators are associated with single-qubit Pauli-$Z$ operators on $d$ data qubits, which we refer to as \emph{base qubits}~\cite{note:change-of-notations}. The remaining ${n-d}$ data qubits are referred to as \emph{parity qubits}. Physical Pauli-$Z$ operators on parity qubits translate to products of logical Pauli-$Z$ operators. Each logical Pauli-$X$ operator has support on one base qubit and ${d-1}$ parity qubits, as illustrated in Fig.~\ref{fig:parity-code}~(a). For the ease of further discussion, we will refer to a string formed by the base qubits as qubit line 0 and to the strings of qubits along Pauli-$X$ operators as qubit lines 1 to $d$. Hence, for a distance-$d$ code, one has $(d+1)$ qubit lines of length $d$, with each pair of qubit lines having exactly one physical qubit in their common support. 

The modified LHZ configuration we consider in this work differs from the universal parity code proposal of Ref.~\cite{PhysRevLett.129.180503} by a deformation applied to the logical operators and stabilizers. Such deformation does not affect the properties of the code on a logical level, however, allows to find one-dimensional \emph{materialized symmetries}~\cite{10048521} of the parity code. In the following sections, we introduce a matching decoder by exploiting underlying conservation laws that arise due to the symmetries of the parity code.

\section{Materialized symmetries of the parity code} \label{sec:symmetries}

We define a materialized symmetry $\mathcal{S}_k$ as a set of stabilizers whose product is the identity operator, 
\begin{equation}\label{eq:symmetry}
    \prod_{G_i \in \mathcal{S}_k} G_i = 1.
\end{equation}
Any symmetry leads to a conservation law for the syndrome~\cite{10048521}. According to  definition~\eqref{eq:symmetry}, the product of all stabilizer measurements along the materialized symmetry must be equal to one. Under any error on physical qubits, a symmetry contains an even number of defects which can be matched in pairs, paving the way to the construction of a matching decoder. For instance, $X$ and $Z$ subsets of the toric code stabilizers form symmetries~\cite{10048521},
$        
\prod_{G_i \in \mathcal{S}_X}G_i
=
\prod_{G_i \in \mathcal{S}_Z}G_i
=
1.
$    
Hence, the syndrome information of the toric code can be decoded by matching defects within two graphs formed by stabilizers belonging to the $X$ and $Z$ symmetries~\cite{10.1063/1.1499754,edmonds_1965}. Similarly, stabilizers of any two colors in the two-dimensional color code form symmetries, allowing for a matching-based restriction decoder~\cite{PhysRevA.89.012317,chamberland2020triangular,Kubica_2023} and its more advanced versions~\cite{PRXQuantum.3.010310}. 

Topological codes can exhibit a richer variety of symmetries when a particular error model is considered. As such, the surface~\cite{PhysRevLett.124.130501,XZZX} and color~\cite{PhysRevLett.133.110601} codes under infinite-bias noise support one-dimensional symmetries, enabling high-threshold symmetry-based MWPM decoding. Since in the bulk of the code the syndrome of the $ZY$ surface code and the syndrome of the parity code behave identically under the action of Pauli-$X$ noise, a decoding principle used in Ref.~\cite{PhysRevLett.124.130501} can be employed for resolving the syndrome along symmetries of the parity code.

We say that a materialized symmetry is linear if it is supported on a
one-dimensional sub-manifold of the lattice. Matching syndromes supported on a linear symmetry is equivalent to decoding a
repetition code. Figure~\ref{fig:parity-code}~(b--g) shows the full set of linear symmetries in a distance-5 parity code. Within a symmetry, each bulk qubit is supported by two stabilizers of the code; hence defects are created in pairs along each symmetry. In contrary, a qubit at the boundary is supported by only one stabilizer along the symmetry perpendicular to the boundary. Hence, when a boundary qubit is flipped, the parity of the syndrome along such a symmetry changes, violating the conservation law of Eq.~\eqref{eq:symmetry}. In order to correctly account for errors along the boundary, we append \emph{virtual stabilizers} to the ends of each symmetry line of Figs.~\ref{fig:parity-code}(b--f). Collectively, these virtual stabilizers form an additional \emph{virtual symmetry} shown in Fig.~\ref{fig:parity-code}(g). As opposed to \emph{real stabilizers} measured during error correction, virtual stabilizers can be freely added during the decoding algorithm to maintain the conservation of the syndrome parity along all of the code symmetries. 

The parity matching decoder exploits syndrome parity conservation under the action of noise. By definition, under any multi-qubit error, defects along the symmetries are created in pairs, and can hence be matched in pairs along the respective symmetries. 
By construction, any pair of symmetries have exactly one common stabilizer, and any stabilizer belongs to exactly two symmetries. Therefore, when combined in pairs along their respective symmetries, each defect of the syndrome will be paired to two defects along the symmetries it belongs two. Collectively, such pairings form a closed contour, which we will refer to as a cluster. 

Any error compatible with the observed syndrome will correspond to the interior of a cluster. To see this, assume a toy example of a parity code containing a single-qubit bit flip. Such an error flips four nearby stabilizers, that is, creates two defects along each of the four symmetries surrounding the qubit. Connecting pairs of defects along each of the symmetries will form a cluster containing a flipped qubit as its interior. Subsequent errors either create new clusters, or deform the existing one, as exemplified in Fig.~\ref{fig:parity-code}~(a) for a two-qubit error. The problem of finding the error compatible with the syndrome therefore reduces to finding the interior of clusters formed by pairing defects along each of the codes symmetries. Examples of such pairings along different symmetries are shown in Figs.~\ref{fig:parity-code}(b)--(g). Combined together, these pairings form a cluster which encloses the error, as depicted in Fig.~\ref{fig:parity-code}(a). The process of matching syndromes along the code symmetries returns the state to the correct codespace and constitutes the first step of the parity decoder. The exact way of matching syndrome defects will affect the performance and complexity of the algorithm. A few concrete choices of a matching decoder are described in detail in Sec.~\ref{sec:decoder}. 

\section{Optimal threshold of the parity code} \label{sec:statistical-properties}

In this section we explicitly prove that the parity code has an optimal threshold of 50\% in the limit of large codes. In the process, we also derive important statistical properties of the code, which will allow us to construct a scalable decoder in the following sections.

We start by introducing a few definitions. In Sec.~\ref{sec:parity-code}, we introduced a qubit line as a size-$d$ subset of the code qubits, such that an operator formed by a product of Pauli-$X$s applied to qubits along a qubit line commutes with the code stabilizers. For convenience, we generalize this notation and refer to the symmetric difference  of any $k$ qubit lines as a $k$-line. Throughout the paper, $k$-lines, errors, and correction are considered to be sets of qubit indices. However, all our derivations can be formulated in terms of operators. As such, a 1-line operator is a product of Pauli-$X$s applied along a 1-line, a $k$-line operator is a product of $k$ 1-line operators, and correction and error operators are product of Paulis applied to the corresponding sets of qubits.

We will denote the number of $k$-lines that a distance-$d$ parity code contains as ${D(d,k) = \binom{d+1}{k}}$. The number of qubits in a $k$-line reads ${W(d,k) = k(d-k+1)}$ and will be referred to as the \emph{weight of a $k$-line}. A $k$-line will be denoted as $L_k^i$, where ${k \in [1, (d-1)/2]}$ and ${i \in [1, D(d,k)]}$. Note that considering $k$-lines with $k > (d-1)/2$ is redundant, since, by construction, any such $k$-line $L_k^i$ is equivalent to some $m$-line $L_m^j$ with ${m = d + 1 - k}$. Finally, in the following two sections we consider codes with large distances ${d \gg 1}$; for brevity, ${d \pm 1}$ is approximated with $d$ throughout the derivations wherever it has no effect on the results. 

Consider a distance-$d$ parity code in the modified LHZ layout of Fig.~\ref{fig:parity-code}. Qubits are subject to independent and identically distributed noise. Assume a $k$-line $L_k^i$. Let $\mathcal{E}$ denote a set of qubits flipped by Pauli-$X$ errors, which we refer to simply as an error. If for a given error $\mathcal{E}$ half or more of the qubits belonging to $L_k^i$ are Pauli-flipped, that is, $|\mathcal{E} \cap L_k^i| \geq W(d,k)/2$, we say that  $L_k^i$ contains a logical error and write $E_k^i(\mathcal{E}) = 1$. Otherwise, we say that $L_k^i$ is error-free and write $E_k^i(\mathcal{E}) = 0$. When it does not cause confusion, we will write simply $E_k^i = 1$ and $E_k^i = 0$ for brevity. 
Below, we prove the following Lemma.

\begin{figure}[t]
\includegraphics[width=0.85\columnwidth]{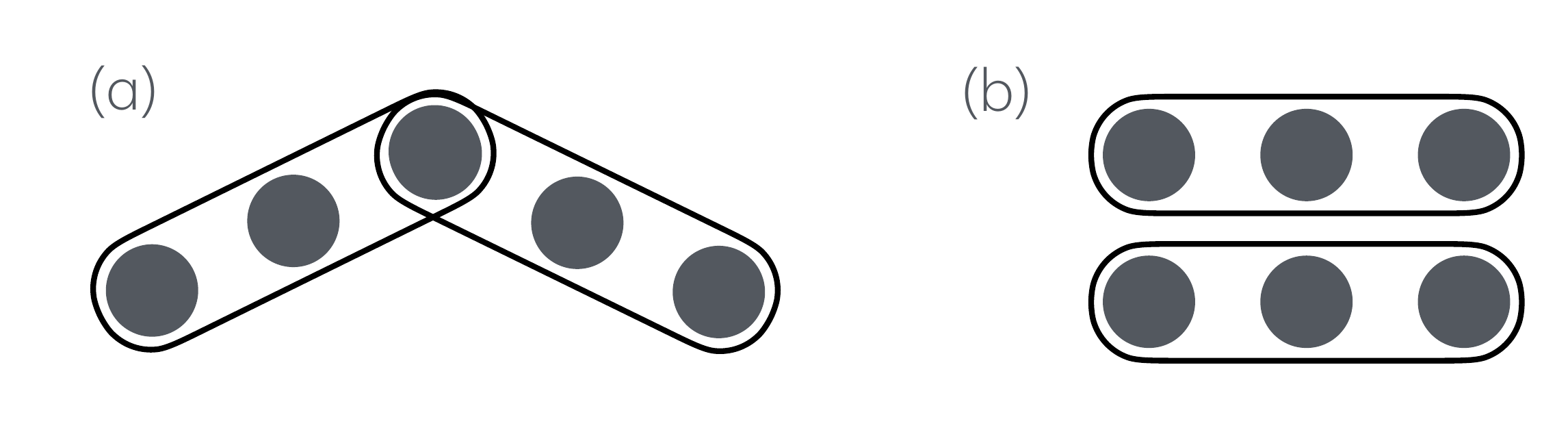}
\caption{\label{fig:toy-example} Two distance-3 codes that (a)~share one common qubit and (b)~have no common qubits.}
\end{figure}

\begin{lemma}\label{lemma:1} Given $k \in [1,d/2]$, the probability that at least one of $D(d,k)$ $k$-lines contains an error is $O{(}e^{-kd}{)}$ for $p<1/2$ and $d \rightarrow \infty$. That is,
\begin{equation}\label{eq:lemma-1}
P
\Big{(}
\sum_{i=1}^{D(d,k)}
E_k^i > 0
\Big{)}
=
O(e^{-kd})
\quad
\forall k \in [1,d/2].
\end{equation}
\end{lemma}

\begin{proof}
The probability that at least one $k$-line contains an error can be written as
\begin{equation}\label{eq:k-line-error}
\begin{aligned}    
    P\Big{(}\sum_{i=1}^{D(d,k)} E_k^i > 0\Big{)} 
    &=
    1 - P\Big{(}\sum_{i=1}^{D(d,k)} E_k^i = 0\Big{)}
    \\&=
    1 - P(E_k^1 = 0)
    P(E_k^2 = 0|E_k^1 = 0)
    ...
\end{aligned}
\end{equation}
where $P(\sum_{i=1}^{D(d,k)} E_k^i = 0)$ is the probability that none of $k$-lines contains an error and $P(E_k^i = 0|E_k^j = 0)$ is the conditional probability that $L_k^i$ is error-free given $L_k^j$ is error-free. Such probabilities are not independent since $L_k^i$ and $L_k^j$ have physical qubits in their common support. However, we note that 
\begin{equation}\label{eq:inequality}
    P(E_k^i = 0|E_k^j = 0)
    \geq
    P(E_k^j = 0).
\end{equation}
Indeed, since any subset of $k$-lines has at least a single qubit in their joint support, the posterior probability of that qubit to contain an error given $L_k^j$ is error-free is $p' \leq p$, where $p$ is the prior error probability of the qubit. As a simple illustration, consider two distance-3 codes in Fig.~\ref{fig:toy-example}. The probability that both codes are error-free when they share a qubit is
\begin{equation}
    P(E_1^2 = 0|E_1^1 = 0)
    P(E_1^1 = 0) = 1 - 6p^2 + 8p^3 + O(p^4), 
\end{equation}
while for two independent repetition codes the probability reads
\begin{equation}
    P^2(E_1^1 = 0) = 1 - 6p^2 + 4p^3 + O(p^4), 
\end{equation}
fulfilling inequality~\eqref{eq:inequality}. 

Therefore, the probability that at least one of $k$-lines of the distance-$d$ parity code contains an error is upper-bounded by the probability that at least one of $D(d,k)$ distance-$W(d,k)$ independent repetition codes contains an error,  
\begin{equation}\label{eq:upper-bound}
\begin{aligned}
    P\Big{(}\sum_{i=1}^{D(d,k)} E_k^i > 0\Big{)} 
    \leq
    1 - \Big{[}1 - P\Big{(}E_k^1 = 1\Big{)}\Big{]}^{D(d,k)}.
\end{aligned}
\end{equation}
The RHS of the equation describes two opposite effects. On the one hand, the probability that a single repetition code contains an error, $P{(}E_k^1 = 1{)}$, is exponentially suppressed with its size $W(d,k)$. On the other hand, the probability that at least one of $D(d,k)$ repetition codes contains an error grows with the code distance $d$. Below we show that the exponential suppression always wins, leading to the result of Eq.~\eqref{eq:lemma-1}.

The probability that a single distance-$W$ repetition code fails at $W \gg 1$ is given by the Chernoff bound 
\begin{equation}
\begin{aligned}    
    P\Big{(}E_k = 1\Big{)}
    \leq
    e^{- \alpha(p)W(d,k)}\,,
\end{aligned}
\end{equation}
with $\alpha(p) > 0$ for $p < 1/2$. 
With this, Eq.~\eqref{eq:upper-bound} simplifies to 
\begin{equation}\label{eq:upper-bound-2}
\begin{aligned}
    P\Big{(}\sum_i E_k^i > 0\Big{)} 
    \leq
    1 - \Big{[}1 - e^{-\alpha(p)W(d,k)}\Big{]}^{D(d,k)}.
\end{aligned}
\end{equation}
Next, we note that $e^{-\alpha(p)W(d,k)}D(k,d) \ll 1$ at fixed $p < 1/2$ and sufficiently large $d \gg 1$. To see this, consider two cases. At $k \ll d$, we can approximate $d!$ and $(d-k)!$ using the Stirling formula, such that
\begin{equation}\label{eq:degeneracy-small-d}
    D(d,k)
    \approx
    \binom{d}{k}
    \approx
    \sqrt{\frac{(d-k)}{d}}
    \frac{d^d}{k!(d-k)^{(d-k)}}
    \approx
    \frac{d^k}{k!}\,,
\end{equation}
and with
\begin{equation}
    W(d,k)
    \approx
    kd,
\end{equation}
we have
\begin{equation}\label{eq:small-k-asymptotic}
    D(d,k)
    e^{-\alpha(p)W(d,k)}
    \approx
    \frac{1}{k!}
    e^{k(\log{d}-\alpha(p) d)}
    =
    O(e^{-kd})
    \ll 
    1\,,
\end{equation}
when $d \gg \log{d}/\alpha(p)$. Therefore, exponential suppression of logical error rate dominates over its growth as we increase $W(d,k)$ and $D(d,k)$ with $d$.

At $k = O(d)$, one can use the Stirling approximation for $k!$ as well and write
\begin{equation}
    D(d,k)
    \approx
    \binom{d}{k}
    \approx
    \frac{(\frac{k}{d})^{-d\frac{k}{d}} (1-\frac{k}{d})^{(1-\frac{k}{d})d}}{\sqrt{2\pi d(1-\frac{k}{d})}}\,.
\end{equation}
Taking the logarithm of both sides, this yields
\begin{equation}
    \begin{aligned}
    \ln{D(d,k)}
    &\approx
    -d
    \Big{[}
    \frac{k}{d} \ln{\frac{k}{d}}
    +
    (1-\frac{k}{d}) \ln{(1-\frac{k}{d})}
    \Big{]}
    \\&-
    \frac{1}{2}
    \ln{2 \pi d (1-\frac{k}{d})}
    =
    dH(\frac{k}{d})
    +
    O(\ln{d}),
    \end{aligned}
\end{equation}
where $H(p) = -p \ln{p} - (1-p)\ln{(1-p)}$ is the binary entropy up to the change in the base of the logarithm. Since this function is upper-bounded by a constant $C = O(1)$, we have 
\begin{equation}
    \frac{\alpha(p)W(d,k)}{dH(k/d)}
    \geq
    \frac{\alpha(p)k(1 - (k-1)/d)}{C}
    \geq
    \frac{\alpha(p)}{2C}k
    \gg 1,
\end{equation}
where we use $2C/\alpha(p) \ll k = O(d)$. Hence,
\begin{equation}
    \alpha(p) W(d, k)
    -
    dH(k/d)
    \approx
    \alpha(p) W(d, k)\,,
\end{equation}
for $d \gg 1/\alpha(p)$. Again, we see that exponential suppression of the logical error rate within a repetition code dominates over the binomial growth due to the number of such codes,
\begin{equation}\label{eq:large-k-asymptotic}
    \begin{aligned}
    D(d,k)
    e^{-\alpha(p)W(d,k)}
    &\approx
    e^{-\alpha(p)W(d,k) + d H(\frac{k}{d})}
    \\&=
    O(e^{-kd})
    \ll 1,
    \end{aligned}
\end{equation}

Therefore, in both regimes $k \ll d$ and $k = O(d)$, we have
$e^{-\alpha(p)W(d,k)}D(d,k) \ll 1$, meaning that the RHS of Eq.~\eqref{eq:upper-bound-2} can be Taylor-expanded, which to the leading order yields
\begin{equation}\label{eq:1st-order-approx}
\begin{aligned}
    P\Big{(}\sum_{i=1}^{D(d,k)} E_k^i > 0\Big{)} 
    \leq
    D(d,k)e^{-\alpha(p)W(d,k)}
    =
    O(e^{-kd}).
\end{aligned}
\end{equation}

From Eqs.~\eqref{eq:small-k-asymptotic}, \eqref{eq:large-k-asymptotic}, and \eqref{eq:1st-order-approx} we conclude that for distance $d \gg \log{d}/\alpha(p)$ large enough, the probability of a logical error taking place due to all $k$-lines for a fixed $k$ is upper-bounded by $O(e^{-kd})$ for any $p<1/2$. \qed
\end{proof}

We can now prove the following Lemma.
\begin{lemma} \label{lemma:2}
The parity code has a code-capacity threshold of 50\% under an optimal decoder. 
\end{lemma}

\begin{proof}
The optimal decoder finds the shortest overall correction compatible with the syndrome. Such a decoder will fail if and only if there exists at least one $k$-line $L_k^i$ such that $E_k^i > 0$. Using $P(A \cup B) \leq P(A) + P(B)$ and the results of Lemma~\ref{lemma:1}, we have
\begin{equation}\label{eq:1st-order-approx-2}
\begin{aligned}
    &P\Big{(}\sum_{k=1}^{d/2}\sum_{i=1}^{D(d,k)} E_k^i > 0\Big{)} 
    \leq
    \sum_{k=1}^{d/2}P\Big{(}\sum_{i=1}^{D(d,k)} E_k^i > 0\Big{)}
    \\&\leq
    \frac{d}{2}
    \textrm{max}_k
    P
    \Big{(}
    \sum_{i=1}^{D(d,k)}
    E_k^i > 0
    \Big{)}
    =
    O(e^{-d}).
\end{aligned}
\end{equation}
Hence at $d \rightarrow \infty$, the probability of a logical error under the optimal decoder reduces exponentially with $d$ for any $p<1/2$. 
\qed
\end{proof}

\section{Optimal vs near-optimal decoding} \label{sec:1-line-decoder}

In the previous section, we used Lemma~\ref{lemma:1} to show that the optimal decoder has a threshold of 50\% in the limit of large codes. The decoder finds the shortest possible correction operator compatible with the syndrome. Such a decoder, however, is impractical, as it requires minimization of the correction operator along exponentially many $k$-lines. In this section, we show that Lemma~\ref{lemma:1} yields an even stronger result, enabling a scalable yet optimal decoding in the regime of large codes. 

Consider a new decoder that we refer to as a \emph{1-line decoder}. It finds a correction compatible with the syndrome and containing less than $\lceil (d+1)/2 \rceil$ of qubits along all 1-lines. However, the decoder does not guarantee that the correction is shorter than half of qubits along all $k$-lines for $k>1$. 
Below, we prove the following Lemma.

\begin{lemma}\label{lemma:3}
    Given an error $\mathcal{E}$ is correctable with the optimal decoder, the probability that the 1-line decoder fails is $O(e^{-d})$ for $d \rightarrow \infty$.
\end{lemma}

\begin{proof}
The proof follows directly from Lemma~\ref{lemma:1}. Consider error $\mathcal{E}$ correctable by the optimal decoder and correction $\mathcal{C}$ returned by the 1-line decoder. 
All 1-lines along $\mathcal{C}$ are error-free, $E_1^i(\mathcal{C}) = 0$ $\forall i \in [0,d]$. Hence, the weight of correction $\mathcal{C}$ is less than half of all of the code's qubits, $|\mathcal{C}| < n/2$. In this regime, $\lceil (d+1)/2 \rceil \approx d/2$ and equation~\eqref{eq:lemma-1} is valid, that is, the probability that $\mathcal{C}$ is longer than half of the qubits along any of $k$-lines for $k>1$ is $O(e^{-kd})$ compared to the probability of an error in any of 1-lines. Since the latter is ruled out by the 1-line decoder, the probability of $\mathcal{C}$ being longer than half along any of $k$-lines is upper bounded by $O(e^{-d})$ in the leading order. 
\qed
\end{proof}

As follows from Lemma~\ref{lemma:3}, at large $d$ the logical error rate of the 1-line decoder is of the same order as the logical error rate of the optimal decoder, $O(e^{-d})$. Therefore, the 1-line decoder achieves optimality at $d \rightarrow \infty$. On the other hand, the 1-line decoder only requires direct minimization of the correction set $\mathcal{C}$ along $d$ 1-lines instead of all $k$-lines. Hence, it's complexity scales much more favorably with the code size. A concrete construction of a 1-line decoder is provided in the next section.

\section{A two-step parity decoder} \label{sec:decoder}

We will now use the statistical and symmetrical properties of the parity code to construct a 1-line decoder for the parity code in two steps. The first step, which we also refer to as the matching step, finds a correction compatible with the observed syndrome by pairing defects along one-dimensional symmetries of the code. Clusters derived by such pairing contain a correction operator as their interior. The second step, which we refer to as post-processing, checks whether the correction returned by the first step is shorter than $\lceil (d+1)/2 \rceil$ along each 1-line $L_1^i$ and, if not, multiplies the correction operator by the corresponding 1-line operator. In conjunction, the two steps effectively constitute the 1-line decoder described in the previous section, that is, find a correction that returns the code into the correct codespace and has a weight less that $\lceil (d+1)/2 \rceil$ along all 1-lines. According to Lemma~\ref{lemma:2} and Lemma~\ref{lemma:3}, the two-step decoder demonstrates the same error correcting capabilities as the optimal decoder and yields the code-capacity threshold of 50\% at large code sizes. Below we describe the two steps in more detail.

\subsection{Step 1: matching}\label{sec:matching-step}

The goal of the matching step is to return the system to the codespace by pairing syndrome defects along their respective symmetries according to some metrics. There exist various possibilities to do so. Indeed, any matching that respects the conservation of the syndrome parity forms closed clusters and hence results in a recovery operator. We will consider three choices of a matching step. 

MWPM aims to find the global minimum of the clusters perimeter during the matching step. Since in the bulk of the code the perimeter of the cluster grows monotonically with its area, the most probable correction operator compatible with the syndrome corresponds to the interior of clusters with the smallest total perimeter, that is, the shortest cumulative pairing along all of the code symmetries, including the virtual symmetry. Such a pairing can be found by executing the MWPM algorithm on a matching graph, which we construct in Appendix~\ref{app:mpwm-graph-ideal} analogously to the surface-code symmetry graph of Ref.~\cite{PhysRevLett.124.130501}. While defects in real stabilizers can only be matched within their respective symmetry, virtual stabilizers can be added or removed freely, which creates efficient interactions between one-dimensional symmetries. Since the MWPM algorithm seeks for the global minimum of the cluster perimeter among all configurations compatible with the observed syndrome, all possible combinations of virtual stabilizers must be considered by the algorithm, leading to the worst-case complexity cubic in the number of nodes~\cite{Kolmogorov2009}, i.e., $O(d^6)$. Throughout the text, the term complexity refers to the worst-case scenario and the total amount of computations required for resolving the syndrome. Parallelisation and average-case runtime will be discussed in Sec.~\ref{sec:complexity}.

\begin{figure}[t]
\includegraphics[width=0.95\columnwidth]{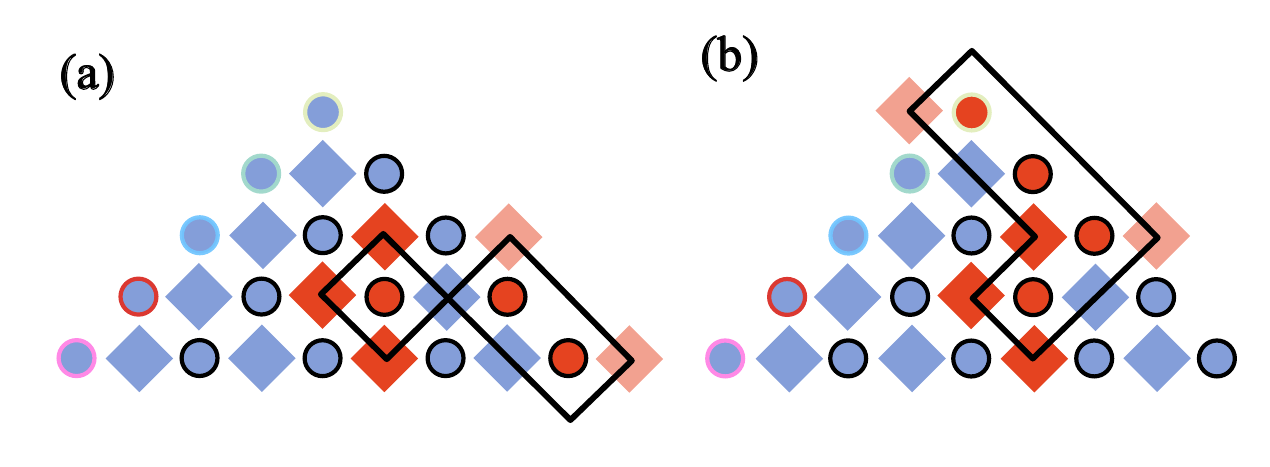}
\caption{\label{fig:post-processing} Ambiguous matching clusters. Clusters in both panels correspond to the same measured syndrome~(bright red plaquettes) but different virtual stabilizers~(pale red plaquettes). Since the two clusters have the same perimeter of length 10, the direct implementation of the MWPM decoder does not distinguish between such degenerate solutions. However, the correction of panel (a) has weight 3, while the solution of panel (b) has weight 4. Executing the post-processing step described in Ref.~\ref{sec:minimization-step} allows to remove the ambiguity and choose the correct, i.e., the lower-weight recovery operation.}
\end{figure}

We note that even though the full MWPM finds the global minimum among all of the syndrome pairings, the matching step can still return a wrong correction operator where the optimal decoder would succeed. Indeed, as mentioned above, the cluster interior grows monotonically with its perimeter in the bulk of the code, however, the boundaries can modify such correspondence. As an example, the two correction operators shown in Fig.~\ref{fig:post-processing} correspond to the same perimeter of the clusters, however, have different weights. Since the MWPM only aims to minimize the perimeter, it might return the higher-weight correction of Fig.~\ref{fig:post-processing}(b) instead of Fig.~\ref{fig:post-processing}(a), and therefore, fail. Furthermore, in Appendix~\ref{app:need-of-post-processing}, we demonstrate a configuration where a shorter perimeter corresponds to a higher-weight recovery operator. As we show in Sec.~\ref{sec:minimization-step}, the second step of the decoder is designed to resolve the correct recovery operator, even when the first step chooses the wrong one.

Next, we consider independent symmetry matching~(ISM). In contrast to MWPM, interactions between one-dimensional symmetries are ignored by treating each symmetry as an independent repetition code. Within each individual 1D symmetry, the decoder finds the lowest-weight correction by either pairing defects along the symmetry, or connecting defects to virtual stabilizers. After this procedure is performed for each of $d$ real symmetries, the remaining virtual stabilizers created at the ends of the symmetry lines are paired along the virtual symmetry. The ISM algorithm is detailed in Appendix.~\ref{app:repetition-graph-ideal}. 

Since the ISM only finds local minima, it can return a matching different from the solution derived with full MWPM. In Appendix~\ref{app:mwpm-vs-independent}, we provide an example of a syndrome which results in different solutions depending on which matching algorithm is used. The worst-case complexity of independent matching is the same as decoding $d+1$ independent distance-$d$ repetition codes, that is, $O(d^2)$ in total.

For completeness, we note that even random matching along each symmetry returns the code to the correct codespace. In this case, pairing defects along each symmetry does not require any minimization, and matching complexity is identical to simply listing all the defects along the symmetries.

\begin{figure}[t]
\includegraphics[width=0.95\columnwidth]{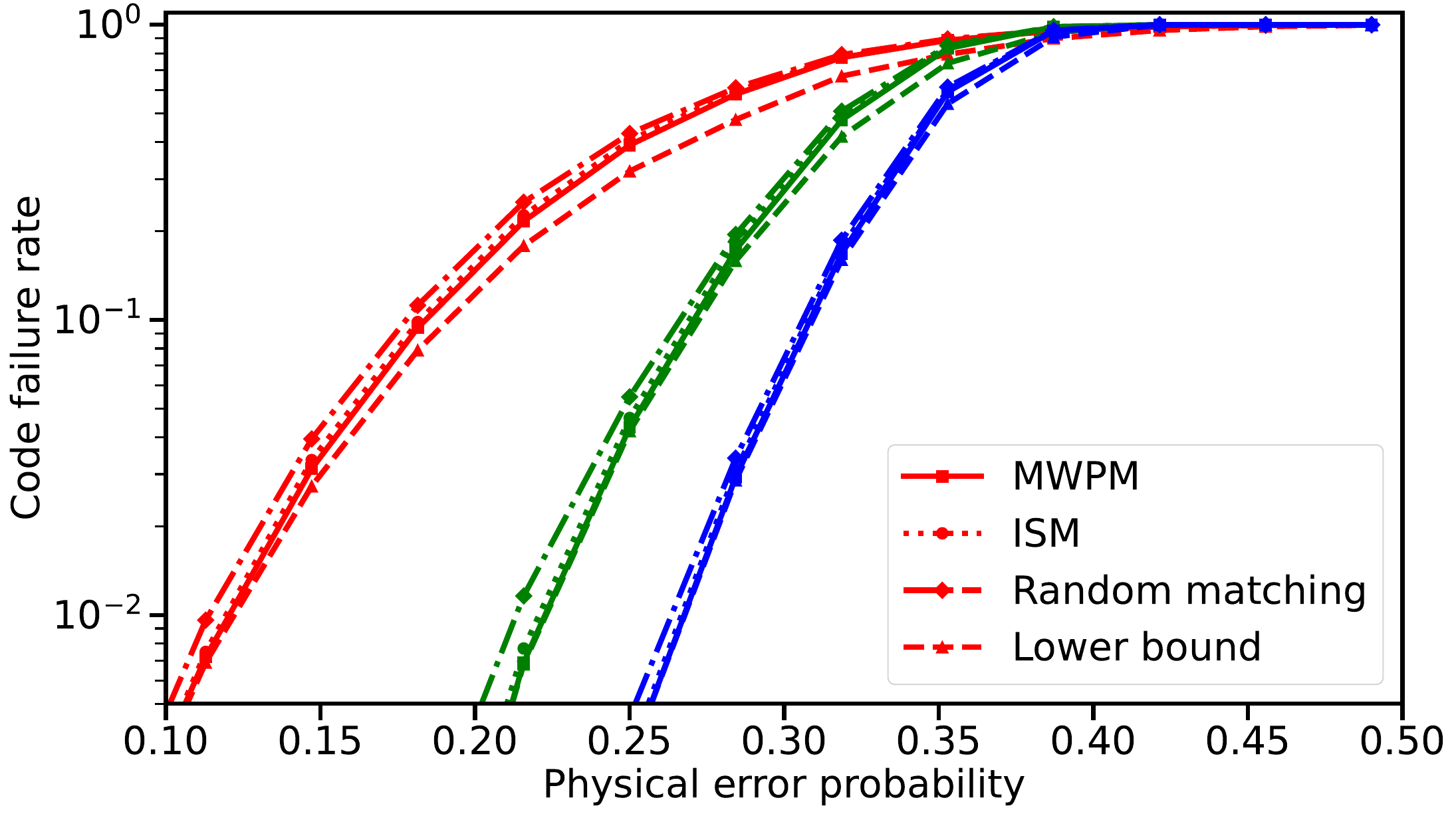}
\caption{\label{fig:matching-benchmark} Logical failure rate of the code, i.e., the probability that at least one of $d$ logical qubits is flipped, versus physical error rate of data qubits. The two-step decoder is used to resolve the syndrome. Three families of curves correspond to parity codes with distances $d=11$~(red), $d=31$~(green), and $d=51$~(blue). Circles, diamonds, and squares correspond to a two-step decoder that uses, respectively, ISM, random matching, and MWPM during the first step of decoding. In all cases, post-processing is applied. Triangles correspond to the probability of a logical error along at least one 1-line $L_1^i$, that is, to a lower-bound failure rate of the optimal decoder.}
\end{figure}

\subsection{Step 2: post-processing}\label{sec:minimization-step}

\begin{figure}[t]
\includegraphics[width=0.95\columnwidth]{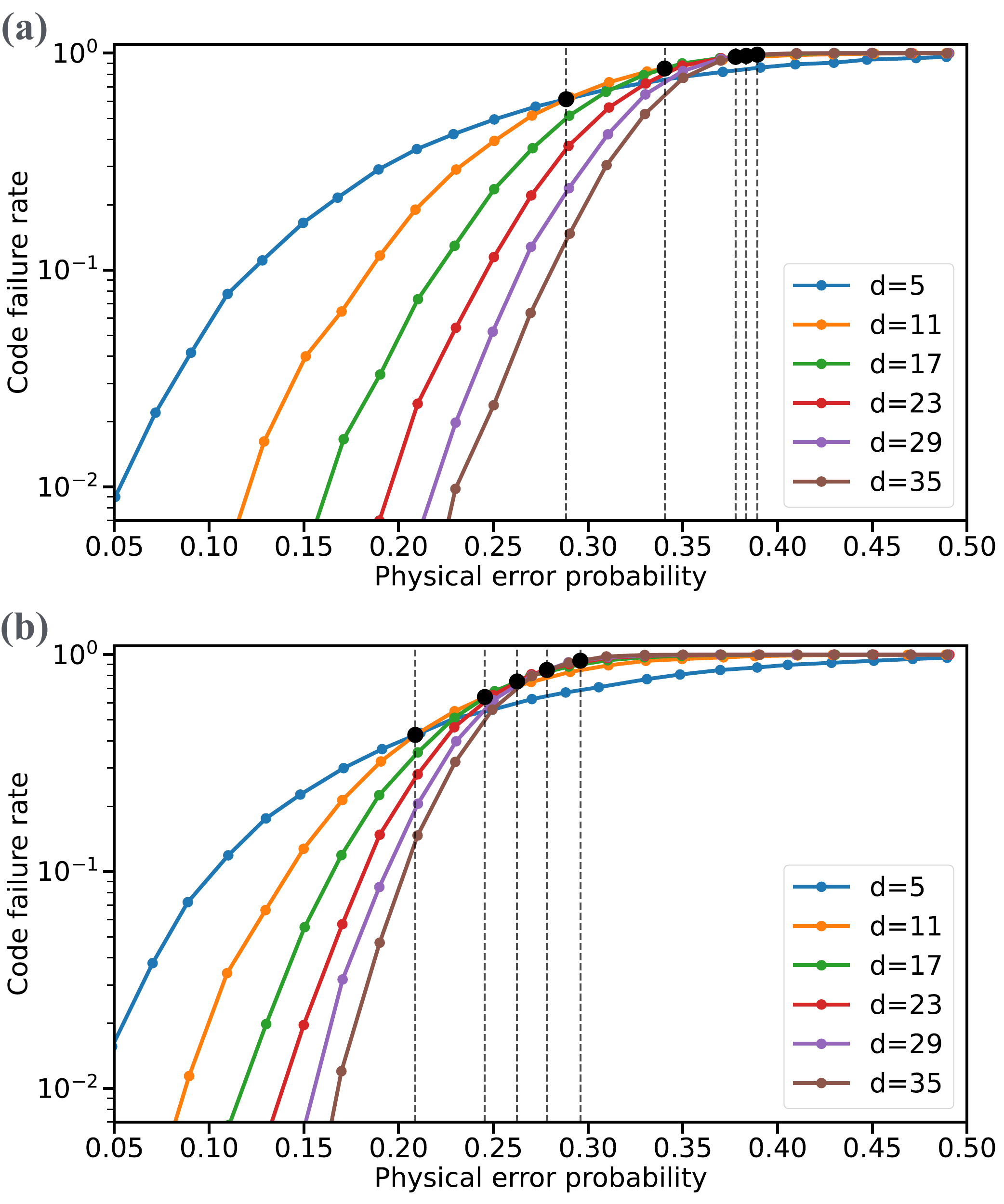}
\caption{\label{fig:LER-ideal} Code failure rate versus physical error rate for different distances $d$~(and, hence, logical system sizes $d$). In both panels we used MWPM decoder during the matching step. Panels~(a) and (b) correspond to decoder operation with and without post-processing, respectively. Dashed lines~(black dots) show the intersection of curves corresponding to $d$ and $d+6$, which shifts towards larger physical error probability as $d$ increases.}
\end{figure}

The second step of the decoder post-processes the output from the first step by ensuring that none of 1-lines contains an error. That is, if correction $\mathcal{C}$ returned by the matching step is such that $|L_1^i \cap \mathcal{C}| > d/2$, $\mathcal{C}$ is replaced by its symmetric difference with $L_1^i$. Using operator notations, a correction operator is multiplied by 1-line operators until the correction is shorter than half qubits long along each 1-line. The second step of the decoder hence reduces to the following algorithm.

\begin{algorithm}[H]\label{alg:pp-ideal}
    \caption{Post-processing}
    \SetAlgoLined
    \KwIn{Set of qubit indices $\mathcal{C}$, output from step 1\;}
    \KwOut{New $\mathcal{C}$ such that $|\mathcal{C} \cap L_1^i| < \lceil (d+1)/2 \rceil \quad \forall i$\;}
    Let $l_i \gets L_1^i\cap \mathcal{C} \quad \forall i$\;
    Let $m \gets \arg\max_{i}|l_i|$\;
    \While{$|l_{\text{m}}| \geq \lceil (d+1)/2 \rceil$}{
        $\mathcal{C} \gets \mathcal{C} \triangle L_1^m$\;
            $l_i \gets L_1^i\cap \mathcal{C} \quad \forall i$\;
        $m \gets \arg\max_{i}|l_i|$\;
    }
    \Return $\mathcal{C}$
\end{algorithm}

Here, $\mathcal{C}$ denotes a set of qubits to be corrected and is initialized to the output from step 1. Each set $l_i$ contains indices of qubits to be flipped along 1-line $L_1^i$, and $m$ labels the largest of the $l_i$s. Inside the loop, $\mathcal{C}$ is replaced by a symmetric difference between sets $\mathcal{C}$ and $L_1^m$, which we denote by $\mathcal{C} \triangle L_1^m$. 
The process repeats until the correction along each 1-line is shorter than $\lceil (d+1)/2 \rceil$ qubits. The algorithm returns a list of qubit indices to be corrected, which is guaranteed to be shorter than $\lceil (d+1)/2 \rceil$ qubits along each 1-line, i.e., returns the output of the 1-line decoder described in Sec.~\ref{sec:1-line-decoder}. As discussed in Sec.~\ref{sec:complexity}, the runtime of the algorithm is $O(d)$ in the sub-threshold regime.

The two-step decoder yields optimal decoding at $d \rightarrow \infty$. In this regime, the choice of a matching algorithm during the first step does not affect the performance of the decoder, since the second step of decoding will turn any input into the optimal 1-line decoder. One can choose either ISM or random matching to minimize the runtime of the entire decoder. Since in this case the first step of the decoder takes the form of $O(d)$ independent distance-$d$ codes, the complexity of the entire algorithm scales as $O(d^2)$
while yielding optimal decoding capabilities for large codes.

\begin{figure}[t]
\includegraphics[width=0.95\columnwidth]{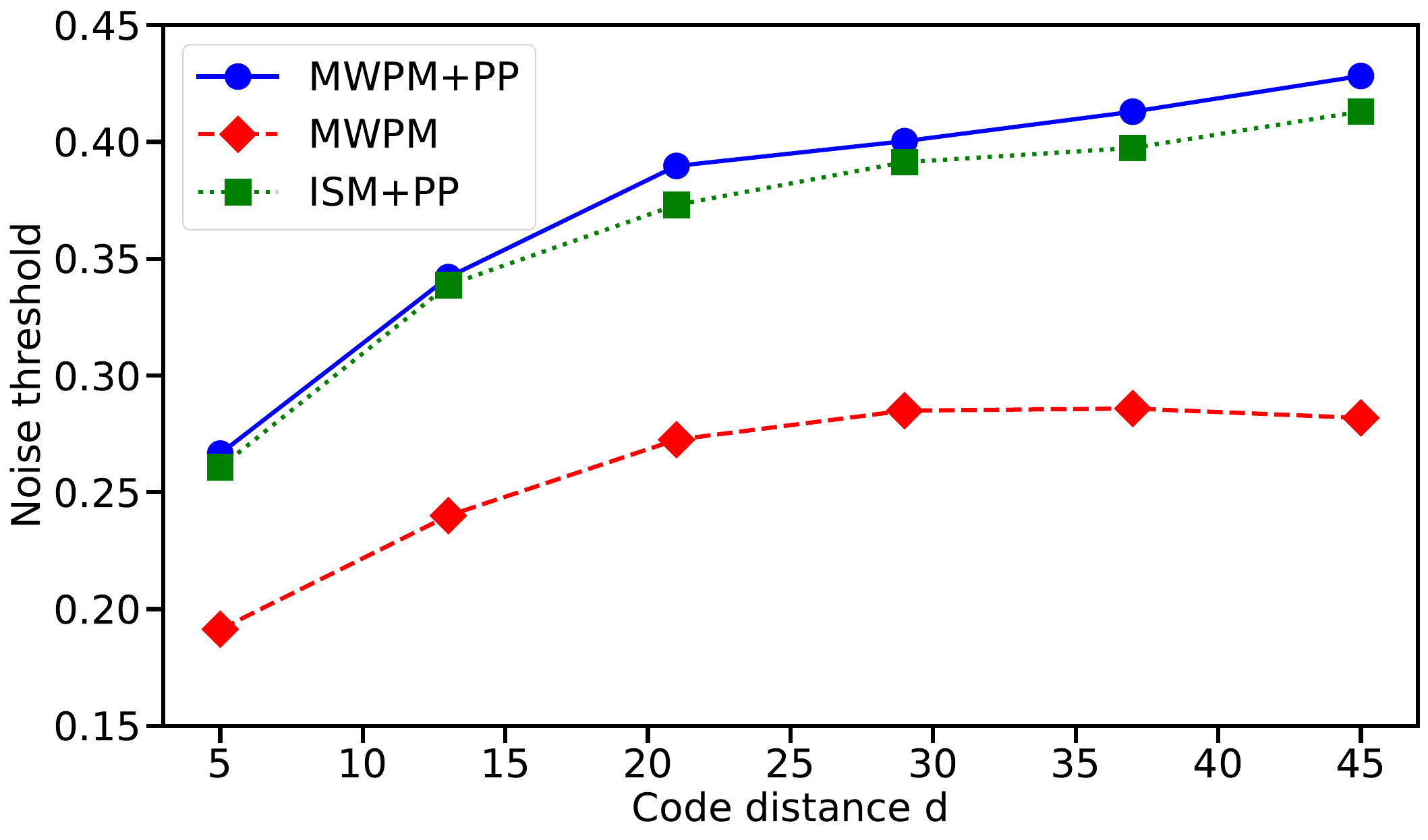}
\caption{\label{fig:thresholds-ideal} Code-capacity threshold of the parity code vs code distance $d$. Post-processing yields a significant improvement compared to the bare matching results. As in Fig.~\ref{fig:matching-benchmark}, the choice of the matching algorithm has little effect on the correction capabilities when post-processing is applied. We define the threshold of a distance-$d$ code as the physical error rate at which curves of Fig.~\ref{fig:LER-ideal} corresponding to $d$ and $d+2$ intersect.
}
\end{figure}

\section{Finite-size codes under ideal measurements} \label{sec:code-capacity}

For finite-size codes the results of Secs.~\ref{sec:statistical-properties} and \ref{sec:1-line-decoder} are not applicable; neither the optimal decoder has a threshold of 50\% nor the 1-line decoder guarantees optimality. In addition, the choice of initial matching might have a significant effect on the decoder performance in such a setting. In this section, we numerically investigate the performance of a two-step decoder for moderate-size codes under ideal measurements. 

Figure~\ref{fig:matching-benchmark} demonstrates the results of Monte-Carlo simulations for various choices of matching during the first step. In the range of simulated code sizes, the three matching algorithms demonstrate similar decoding performance. We also compare the failure rate of the two-step decoder to the lower-bound failure rate, which corresponds to the probability of having an error along at least one 1-line. As shown in the figure, both the MWPM and ISM saturate the lower bound in the regime of low physical error rates. Such behaviour reflects the fact that the code failures occur predominantly due to errors in 1-lines. As we increase the code size, the saturation region spans larger error rates, indicating the optimality of the two-step decoder for the entire region ${0<p<1/2}$ at $d \rightarrow \infty$. We note that, at the cost of a worse complexity scaling, the accuracy of the decoder can in principle be increased by considering also $m$-lines with ${m>1}$ in post-processing, that is, by using an $m$-line instead of a 1-line decoder. However, Fig.~\ref{fig:matching-benchmark} demonstrates that considering only 1-lines saturates the lower bound even for relatively small codes. Hence, in this work, we focus solely on a 1-line decoder.

\begin{figure}[t]
\includegraphics[width=0.85\columnwidth]{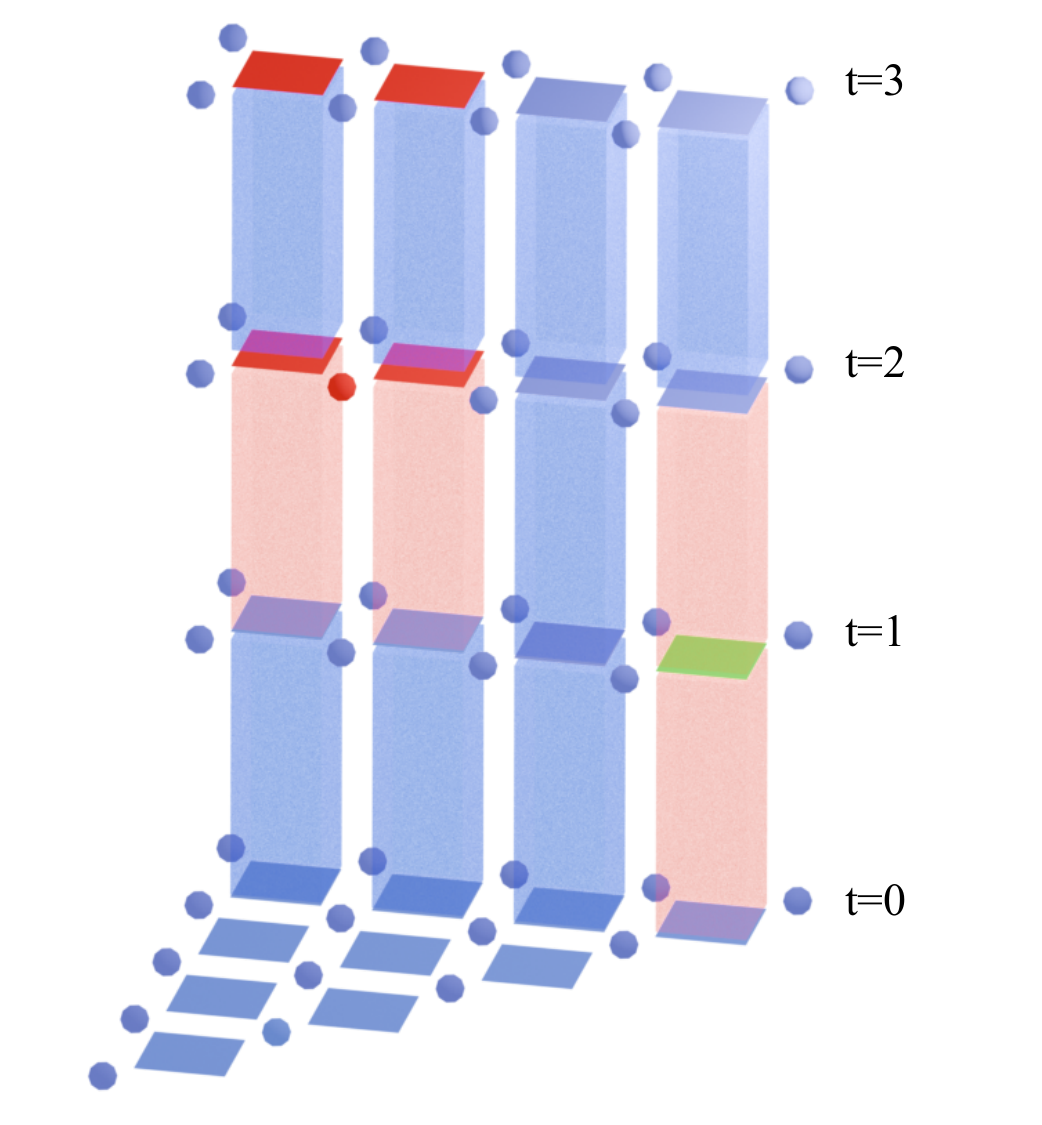}
\caption{\label{fig:3d-symmetry} A spacetime symmetry of a distance-5 parity code under faulty measurements. Time slice $t=0$ shows the full distance-5 code, while other time slices only show qubits and stabilizers belonging to one symmetry. Blue plaquettes correspond to $+1$ stabilizer measurements. Red and green plaquettes correspond to $-1$ stabilizer measurements flipped by qubit and measurement errors, respectively. Note that here we use different colors merely to demonstrate how the syndrome behaves under different types of errors; in real experiments and in our simulations green and red plaquettes are indistinguishable. Cuboids show spacetime stabilizers, i.e., parities between subsequent stabilizer measurements, with red cuboids being spacetime defects. A data qubit error~(red sphere) creates a pair of spacetime defects along a spatial dimension of the symmetry. A measurement error~(green plaquette) at $t=1$ plane creates a pair of spacetime defects along the temporal dimension of the symmetry. Hence, the parity of defects is conserved within the shown subset of spacetime stabilizers. The subset is therefore a symmetry of the code.}
\end{figure}

Figure~\ref{fig:LER-ideal} reflects the importance of the post-processing step. While the MWPM algorithm can be used without any further post-processing, its performance in this case is significantly reduced due to the cluster ambiguity exemplified in Fig.~\ref{fig:post-processing}. Note that in both panels of the figure, the intersection of curves shifts towards higher physical error rate $p$ as $d$ is increased. Identifying a threshold of a distance-$d$ code with a physical error rate at which the curves corresponding to distances $d$ and $(d+2)$ intersect, we plot threshold curves in Fig.~\ref{fig:thresholds-ideal} and find a significant improvement of the threshold value when post-processing is applied. As expected, the threshold grows monotonically with distance $d$ and, when post-processing is used, exceeds 40\% for moderate-sized codes. This behaviour agrees with the analytical results of Secs.~\ref{sec:statistical-properties} and \ref{sec:1-line-decoder}, where we have shown saturation of the one-line decoder threshold at $p_{\textrm{th}}=1/2$ in the limit of large codes. 

\begin{figure}[t]
\includegraphics[width=0.9\columnwidth]{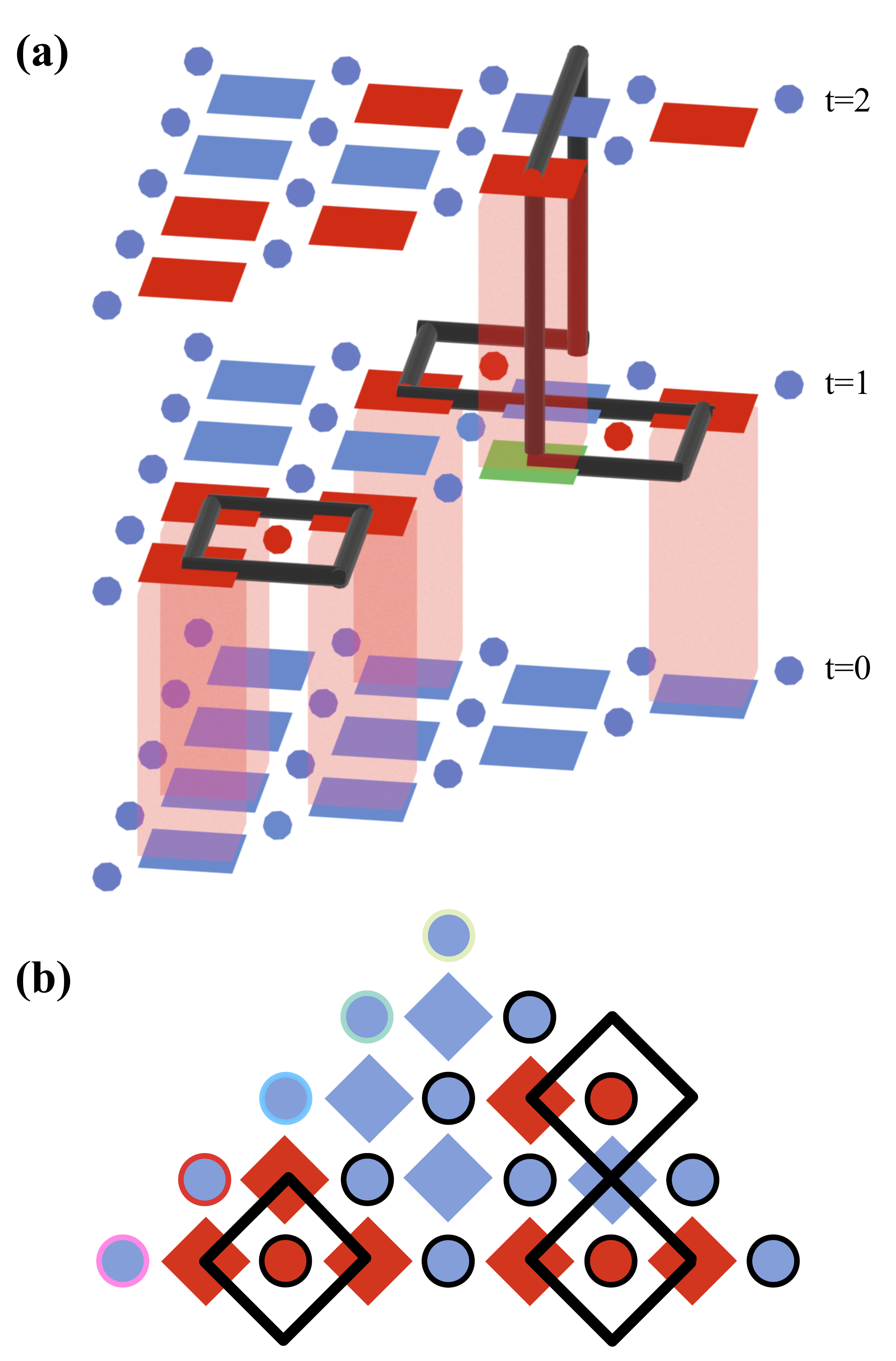}
\caption{\label{fig:3d-matching} Spacetime syndrome matching in a distance-5 LHZ code under noisy measurements. (a)~Errors in data qubits~(red circles) and a single measurement error~(green plaquette) form spacetime defects~(red cubiods). Only the defects are shown, error-free spacetime stabilizers are omitted to avoid cluttering. Tops of red cuboids are paired~(black curves) during the matching step along each of the $(1+1)$-dimensional symmetries, forming closed clusters. (b)~Clusters in $(2+1)$ dimensions are projected onto a 2-dimensional parity code at time slice $t=1$. Errors in data qubits correspond to the interior of the projected clusters.}
\end{figure}

\begin{figure*}[t]
\includegraphics[width=0.95\textwidth]{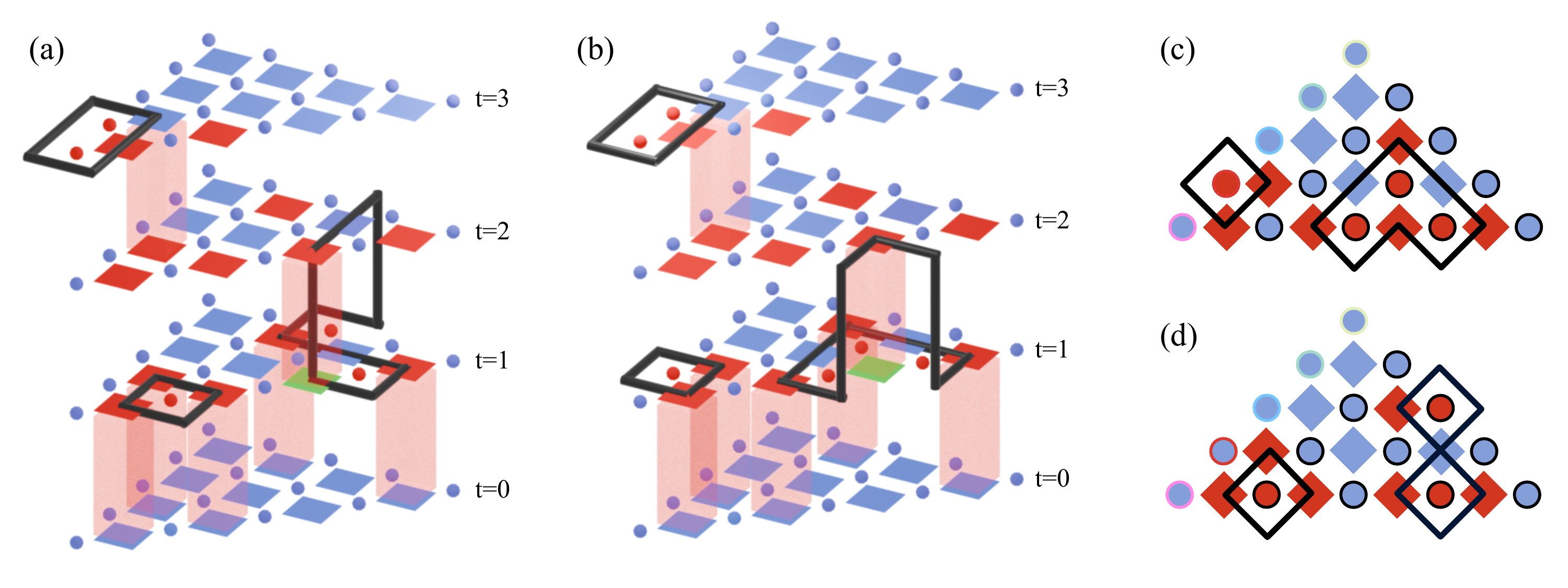}
\caption{\label{fig:3d-post-processing} Post-processing of a spacetime syndrome. 
Panels (a) and (b) show clusters corresponding to the same syndrome. While the total cluster perimeters of two configurations are identical, the weight of spacetime correction $\mathcal{C}_{\textrm{3D}}$ in panel~(a) is smaller. In contrary, the cumulative error $\mathcal{C}$ calculated according to Eq.~\eqref{eq:C-projected} is smaller for configuration (b). Since the most probable case corresponds to the lowest-weight error across $d$ measurement rounds, the goal of post-processing is to minimize $\mathcal{C}_{\textrm{3D}}$. First, project each cluster onto a single time slice by majority voting and find the clusters interior. As an example, panel (c)~shows the projected three-dimensional cluster of panel~(b) onto time slice $t=1$. Next, run the post-processing algorithm~\ref{alg:pp-ideal} within each time slice. This will transform operator of panel~(c) into the operator of panel~(d). When collected together, the post-processed correction operators within each time slice will yield the configuration of panel~(a). Hence, the post-processing minimized the weight of $\mathcal{C}_{\textrm{3D}}$. The final correction can be calculated by finding the cumulative effect of $d$ stabilizers measurements according to Eq.~\eqref{eq:C-projected}.}
\end{figure*}

\section{Noisy measurements} \label{sec:faulty-measurements}

We now turn to the case where ancilla measurements are unreliable and may result in incorrect syndromes that violate the parity conservation law. Following the
approach of Ref.~\cite{PhysRevLett.124.130501}, we treat errors in data and ancilla qubits on equal footing and construct a new spacetime syndrome in $(2+1)$ dimensions, where the temporal dimension is formed by repeated stabilizer measurements. The parity conservation laws can be recovered along new symmetries within this higher-dimensional space. 

\subsection{Spacetime symmetries}

The new spacetime symmetries are formed by $(1+1)$-dimensional subspaces of a $(2+1)$-dimensional spacetime. An example of a spacetime symmetry for the code under faulty measurements is shown in Fig.~\ref{fig:3d-symmetry}. Each time slice $t$ corresponds to a single round of stabilizer measurements. The spacetime syndrome is formed by parities between subsequent stabilizer measurements, with defects corresponding to a change in two measurement outcomes. If an error occurs in a data qubit, it creates a pair of spacetime defects along the spatial dimension of a ${(1+1)}$-dimensional symmetry. If a stabilizer measurement is faulty, a pair of spacetime defects is created along the temporal dimension. Together, data and measurement errors create a spacetime syndrome in ${(2+1)}$ dimensions which respects the parity conservation laws within each of $d$ $(1+1)$-dimensional symmetries. As in the case of ideal measurements, we add virtual spacetime symmetries to preserve the syndrome parity at both the spatial and the temporal boundaries of the code. 

\subsection{Matching in $(2+1)$ dimensions}

Since the parity conservation laws are respected along new higher-dimensional symmetries, the clustering technique used in the ideal-measurement case can be employed to find the correction operator in the noisy measurement setting too. Recall that data qubit errors create and deform clusters while preserving the parity conservation within each symmetry. Similarly, when faulty measurements are present, a single measurement error will deform the cluster along temporal dimension of the symmetry. Since both data and measurement errors respect spacetime symmetries, connecting defects along each symmetry will form a closed contour. One can find the correction operator by determining the spatial interior of the cluster. 

Figure~\ref{fig:3d-matching} demonstrates an example of two clusters derived by matching. A cluster formed solely by data qubit errors belongs to a single time slice, while the combinations of qubit and measurement errors create a three-dimensional cluster across time slices. To find the correction, we first project each cluster to a single plane of the parity code, which removes information about where and when the measurement errors took place. Figure~\ref{fig:3d-matching}~(b) shows a projection of both clusters to time plane $t=1$. Within each plane, we find indices $\{q\}$ of qubits enclosed by projected clusters, similarly to the ideal-measurement case. Combined together, qubit indices $\{q\}$ at each time slice $t$ form set $\mathcal{C}_{\textrm{3D}}=\{(t,q)\}$ in $(2+1)$-dimensional spacetime, which contains information about qubit errors that took place at each round of stabilizer measurements and constitutes the output from the first step of decoding. 

As in the case of ideal measurements, one can make various choices of matching algorithms in $(2+1)$ dimensions. The full MWPM finds the shortest-perimeter cluster configuration among all possible combinations of virtual stabilizers. Since for a distance-$d$ parity code the syndrome measurement has to be repeated $d$ times, the number of defects is $N_{\textrm{d}} = O(d^3)$ and the worst-case complexity of the MWPM scales as $O(N_{\textrm{d}}^3) = O(d^9)$. Alternatively, we can implement ISM by matching defects independently within each of $d$ $(1+1)$-dimensional symmetries. As in the code-capacity case, the ISM algorithm does not guarantee a global minimum, but has a complexity that scales significantly better than the full MWPM. Namely, the ISM matching step reduces to running MWPM on $d$ independent matching graphs, each corresponding to a spacetime symmetry and containing $O(d^2)$ vertices. The total worst-case complexity is hence $O(d^7)$. We provide details on the construction of matching graphs for MWPM
and ISM in Appendices~\ref{app:mwpm-graph-3d} and \ref{app:repetition-graph-3d}, respectively. 

\subsection{Correction and post-processing under noisy measurements}

\begin{figure}[t]
\includegraphics[width=0.95\columnwidth]{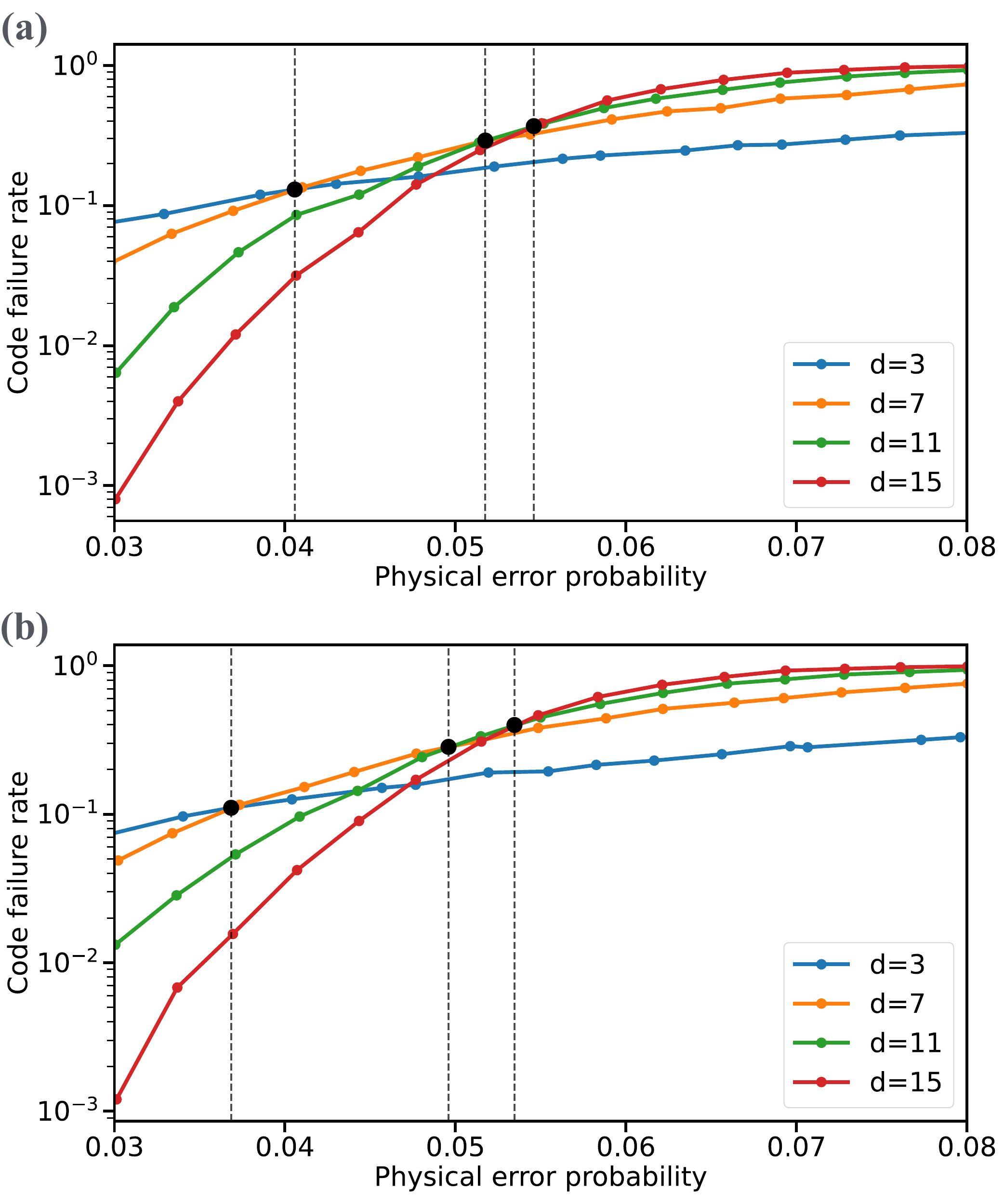}
\caption{\label{fig:LER-FT}Code failure rate versus physical error rate for different distances $d$. The measurement error rate is assumed identical to the data qubit rate. The data is shown for MWPM (a)~with and (b)~without post-processing. Analogous results for ISM are provided in Fig.~\ref{fig:si-comparison-1} of Appendix~\ref{app:data-for-ft-matching}. In all cases, the threshold is improved by post-processing for small-distance codes, however, the effect is more prominent when ISM is in use.}
\end{figure}

Under noisy measurements, the matching step returns the spacetime correction $\mathcal{C}_{\textrm{3D}}$, containing information about which qubits have been Pauli-flipped at each round of stabilizer measurements. To determine the final correction to be applied after $d$ measurement rounds, we calculate the \emph{cumulative correction} $\mathcal{C}$ by calculating a marginal sum of flips modulo 2 for each qubit,
\begin{equation}\label{eq:C-projected}
    \mathcal{C}
    =
    \Big{[}\sum_{t=0}^{d-1}
    \mathcal{C}_{\textrm{3D}}
    \Big{]}\mod{2}.
\end{equation}
This reflects the fact that a qubit flipped twice during $d$ measurement rounds does not need to be corrected.

Finally, we tailor the post-processing procedure developed in Sec.~\ref{sec:minimization-step} to the noisy-measurements setup. The most likely error corresponds to the lowest-weight physical error configuration compatible with the spacetime syndrome. Hence, our aim is to minimize the weight of the spacetime correction $\mathcal{C}_{\textrm{3D}}$, not the cumulative correction $\mathcal{C}$. As an example, assume a configuration shown in Fig.~\ref{fig:3d-post-processing}. Panels (a) and (b) correspond to two configurations of clusters compatible with the same syndrome. Since the total perimeters of the clusters in the two panels are identical, the matching step of the decoder can return either of the two. The total weight of correction $\mathcal{C}_{\textrm{3D}}$ enclosed by clusters in (a) is smaller than in (b), therefore, the former corresponds to the most-probable error and has to be returned by the decoder. In contrary, the weight of the cumulative correction calculated according to Eq.~\eqref{eq:C-projected} is smaller for configuration~(b). Therefore, executing the post-processing step on the cumulative correction in this case transforms the most likely configuration into the less likely one, which increases the logical failure rate. 

We aim to reduce the weight of the full correction operator $\mathcal{C}_{\textrm{3D}}$, which can be achieved by applying post-processing within each of $d$ measurement cycles. That is, for each time slice in $\mathcal{C}_{\textrm{3D}}$, the post-processing step is executed independently, which reduces the weight of the correction within each time slice, and, consequently, the total weight of the spacetime correction $\mathcal{C}_{\textrm{3D}}$. In the example of Fig.~\ref{fig:3d-post-processing}, this process transforms the higher-weight configuration of panel~(b) into the shortest weight configuration of panel~(a). Finally, we apply the correction to the set of qubits calculated in Eq.~\eqref{eq:C-projected}. 

Post-processing under faulty measurements requires execution of Algorithm~\ref{alg:pp-ideal} for each of $d$ time slices, yielding the worst-case complexity $O(d^2)$. Since the complexity of the matching step scales as $O(d^9)$ for MWPM, and as $O(d^7)$ when ISM is in use, adding post-processing has very little effect on the full decoding complexity. 

\subsection{Finite-size codes under noisy measurements}\label{sec:finite-size-ft}

We numerically benchmark the decoder under noisy measurements with and without post-processing. As per standard simulation methods with noisy measurements~\cite{qecsim}, we assume the last round of measurements is ideal to guarantee that the state is returned to the codespace after all spacetime defects have been matched. Figure~\ref{fig:LER-FT}~(Fig.~\ref{fig:si-comparison-1}) demonstrates the logical failure rate of the code versus the physical error probability for the case of the MWPM~(ISM) algorithm used during the matching step. While for MWPM, post-processing has little effect on the decoding capabilities, improvement is significant when ISM is used for decoding, both in terms of the threshold and sub-threshold scaling. As such, Figure~\ref{fig:thresholds-ft} shows the fault-tolerant thresholds as a function of the code size. As in the case of ideal measurement, the threshold increases monotonically with the code distance $d$ for any matching algorithm, both with and without post-processing. For ISM, post-processing noticeably improves the threshold value almost to what is achievable with full MWPM. The effect is especially prominent for moderately small codes.

Similarly, post-processing significantly improves the decoding capabilities of the ISM decoder in the sub-threshold regime. As shown in Fig.~\ref{fig:ft-benchmarks}(a), post-processing yields exponential suppression of logical error rates compared to bare ISM results as we decrease the probability of physical errors below the threshold. A similar exponential suppression also occurs as the code distance is increased in the sub-threshold regime, as demonstrated in Fig.~\ref{fig:ft-benchmarks}(b).

\section{Runtime and parallelization} \label{sec:complexity}

In the previous sections, we briefly discussed the complexity of the two-step decoder. Here, we provide a more detailed runtime analysis for various regimes and configurations of the decoder. 

\subsection{Runtime under ideal measurements}

Let us consider the case of ideal measurements. In Sec.~\ref{sec:matching-step}, we introduced two algorithms for the matching step. MWPM finds the configuration with the smallest perimeter among all clusters compatible with the syndrome in time cubic in the number of nodes $N_{\textrm{d}}$ in the graph. For the worst-case scenario with $N_{\textrm{d}} = O(d^2)$, this yields a runtime scaling as $O(d^6)$. Alternatively, ISM finds the shortest path within each of $d$ symmetries independently and hence can be parallelized. Since matching in a single repetition code can be done in time scaling as $O(d)$, matching the entire syndrome also scales as $O(d)$ when executed on $d$ nodes in parallel.

The post-processing described by Algorithm~\ref{alg:pp-ideal} requires the calculation of the overlap between the correction set and each of 1-lines, $|\mathcal{C} \cap L_1^i|$, which can be executed in $O(d)$ steps using $d$ nodes in parallel. Updating the correction set as $\mathcal{C} \triangle L_1^m$ is a modulo-2 summation that also takes $O(d)$ steps. Algorithm~\ref{alg:pp-ideal} executes the two operations until convergence, that is, until the overlap $|\mathcal{C} \cap L_1^i|$ is shorter than half of the number of qubits along all 1-lines. This potentially increases the complexity of the algorithm. However, our numerical simulations show that in the sub-threshold regime, only a small number of such cycles has to be executed on average. As such, Fig.~\ref{fig:number-of-cycles} shows an average number of cycles required for the convergence as a function of code distance. For a fixed sub-threshold error probability $p$, the number of cycles monotonically reduces after the initial growth. As follows from the figure, one would require less than one cycle of post-processing on average for a realistic setting, that is, when the physical error rate is well below 50\%. We can therefore upper-bound the runtime of the post-processing algorithm by $O(d)$. Hence, when executed on $d$ computational nodes in parallel, the full two-step decoder has a runtime scaling as $O(d)$ when ISM is used during the first step, and as $O(d^6)$ when MWPM is used. 

\begin{figure}[t]
\includegraphics[width=0.95\columnwidth]{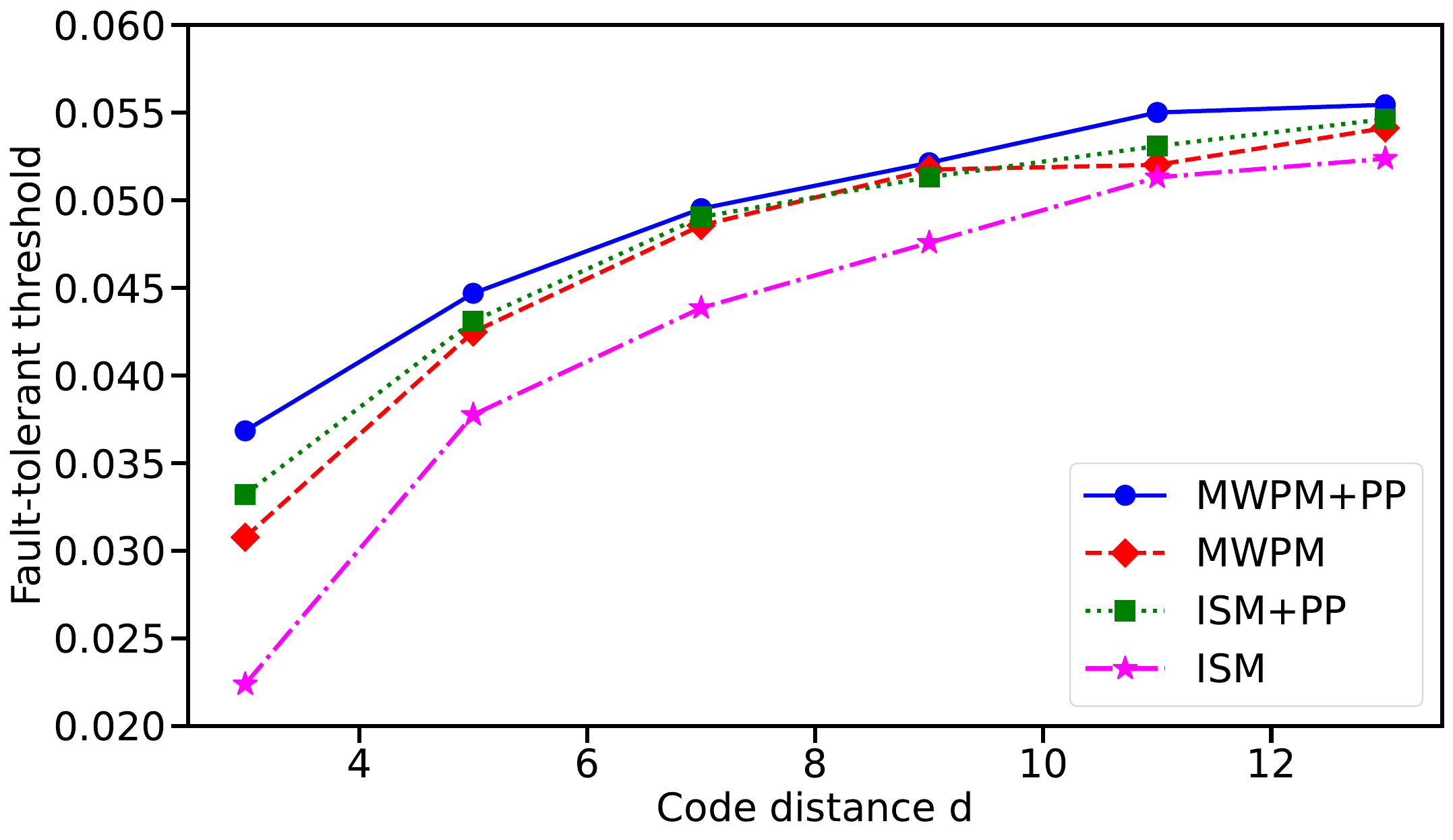}
\caption{\label{fig:thresholds-ft} Fault-tolerant threshold of the parity code vs code distance $d$.}
\end{figure}

\subsection{Runtime under faulty measurements}

Under faulty measurements, the syndrome graph contains $N_{\textrm{d}} = O(d^3)$ vertices and can be matched using MWPM in $O(N_{\textrm{d}}^3) = O(d^9)$ steps. With ISM, the MWPM algorithm is executed independently for each of $d$ graphs representing a $(1+1)$-dimensional symmetry. This can be executed in parallel on $d$ nodes with a runtime scaling of $O(d^6)$. 

Post-processing can as well be applied in parallel within each of $d$ time slices, yielding the runtime of the post-processing step scaling as $O(d^2)$ when executed in parallel on $d$ computational nodes. Furthermore, as discussed in the previous section, post-processing within each time slice can be further parallelized and executed in $O(d)$ on $d$ nodes. This yields a post-processing runtime of $O(d)$ when $d^2$ computational nodes are used. Hence, for any choice of matching algorithms considered in this work, the runtime of the entire decoder is dominated by the matching step, while executing the post-processing step has vanishingly small effect on the decoder runtime for large code distances $d$.

We note that there exist matching algorithms whose runtime scales more favorably than that of MWPM. As such, the worst-case runtime of the Union-Find~(UF) decoding algorithm scales almost-linearly in the number of nodes~\cite{Delfosse_2021}. Applied to the entire $(2+1)$ dimensional graph, this yields the complexity $O(d^3\alpha(d^3)))$ with $\alpha(d^3) \leq 3$ for practical cases. Implementing UF within each of $(1+1)$ dimensional symmetries in parallel, the first step of the decoder can be executed in $O(d^2\alpha(d^2)))$. Since the post-processing step also runs in $O(d^2)$---or even in $O(d)$ when executed on $d^2$ nodes---the worst-case runtime of the entire decoder scales almost linearly in the code size $n \approx d^2/2$.

\begin{figure}[t]
\includegraphics[width=0.95\columnwidth]{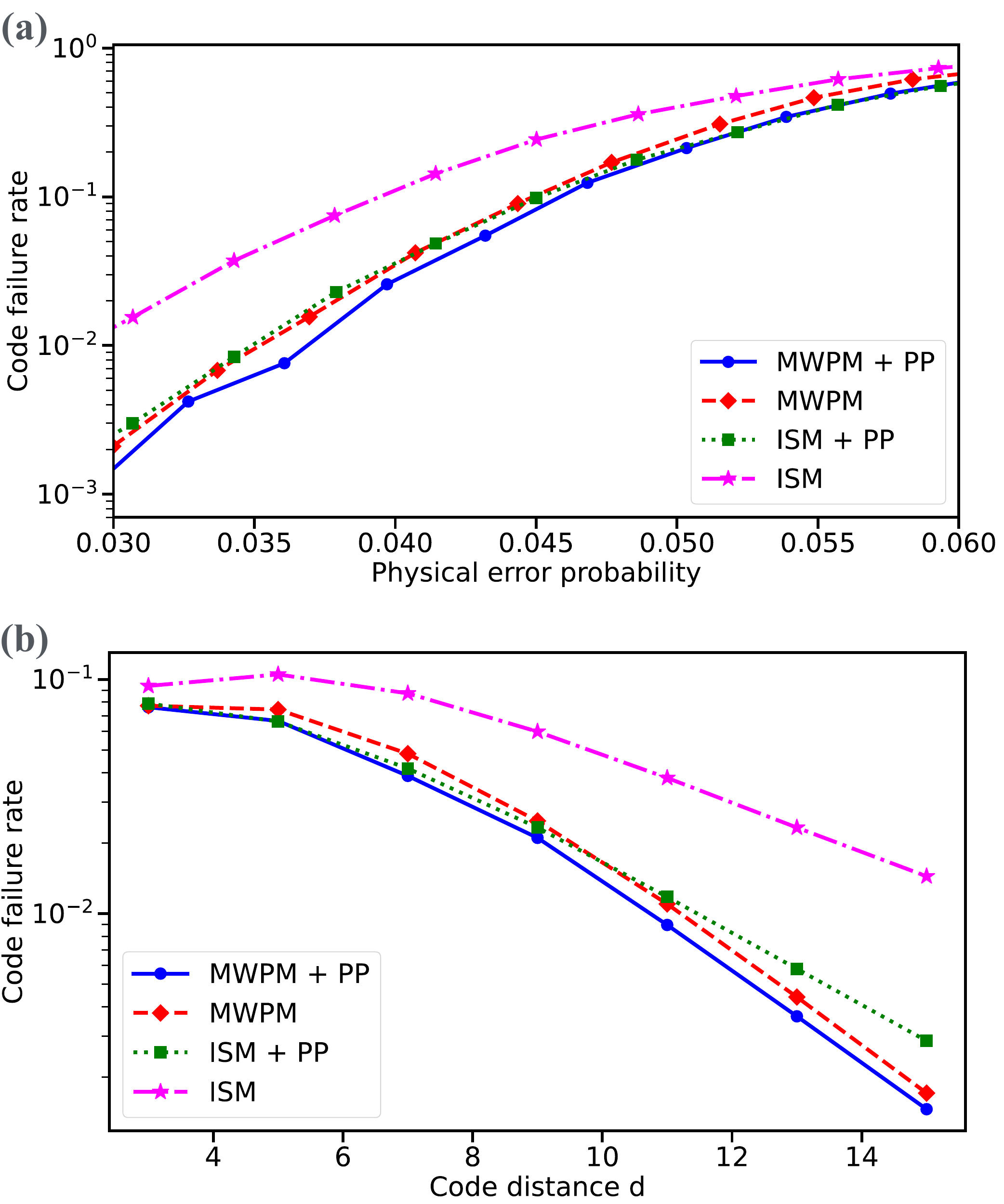}
\caption{\label{fig:ft-benchmarks} Effect of post-processing on logical failure rates for different configurations of the decoder. Panel~(a) shows logical failure rates of the code as a function of the physical error probability $p$ below the threshold at distance $d=15$. In panel~(b), the physical error rate is fixed at $p=3\%$. We observe little effect of post-processing when used in conjunction with MWPM. In contrary, when ISM is used for matching, post-processing allows to suppress logical failure rate exponentially as we increase $d$ or decrease $p$ below the threshold.}
\end{figure}
\subsection{Average-case runtime}

So far we have discussed the complexity in the worst-case scenario, with $O(N_{\textrm{d}}) = O(d^3)$ spacetime stabilizers being flipped in the setting with faulty measurements. However, in practical cases, one is typically interested in performing quantum computation below the fault-tolerant threshold. In this regime, errors and, subsequently, spacetime defects occur with small probability, yielding a sparse spacetime syndrome graph. Such graphs can be efficiently decoded using fast implementations of the MWPM algorihtm, such as sparse blossom~\cite{Higgott_2025} and fusion blossom~\cite{10313859}. In particular, the former runs in $O(N_{\textrm{d}}^{1.17})$ for a graph with $N_{\textrm{d}}$ nodes when a circuit-level noise of strength $0.1\%$ is assumed. Hence, well below the threshold, the modern software implementations of the MWPM algorithm scale almost linearly with the number of nodes. As an example, sparse blossom decodes a distance-17 surface code circuit in under one
microsecond per round on a single core, matching the rate at which syndrome data is generated on a superconducting quantum processor. Note that a full matching graph of a distance-$d$ surface code contains approximately as many vertices as a graph corresponding to one spacetime symmetry of distance-$d^{3/2}$ parity code. Hence, when executing ISM with sparse blossom algorithm for each spacetime symmetry in parallel under the same noise model, we can expect the parity code with distance $d \approx 17^{3/2} \approx 70$ to be decodable in approximately one microsecond. 

Finally, we note that the first step of the decoder does not necessarily have to be based on matching. There exist a plethora of other, non-matching, decoding algorithms, such as belief propagation~\cite{wolanski2025ambiguityclusteringaccurateefficient} or approximate maximum likelihood~\cite{PhysRevA.90.032326,PhysRevX.9.041031}. Any algorithm returning a spacetime recovery similar to the one of Fig.~\ref{fig:3d-matching}(a) can be used in conjunction with post-processing. We leave the construction of alternative decoders to future work.

\begin{figure}[t]
\includegraphics[width=0.92 \columnwidth]{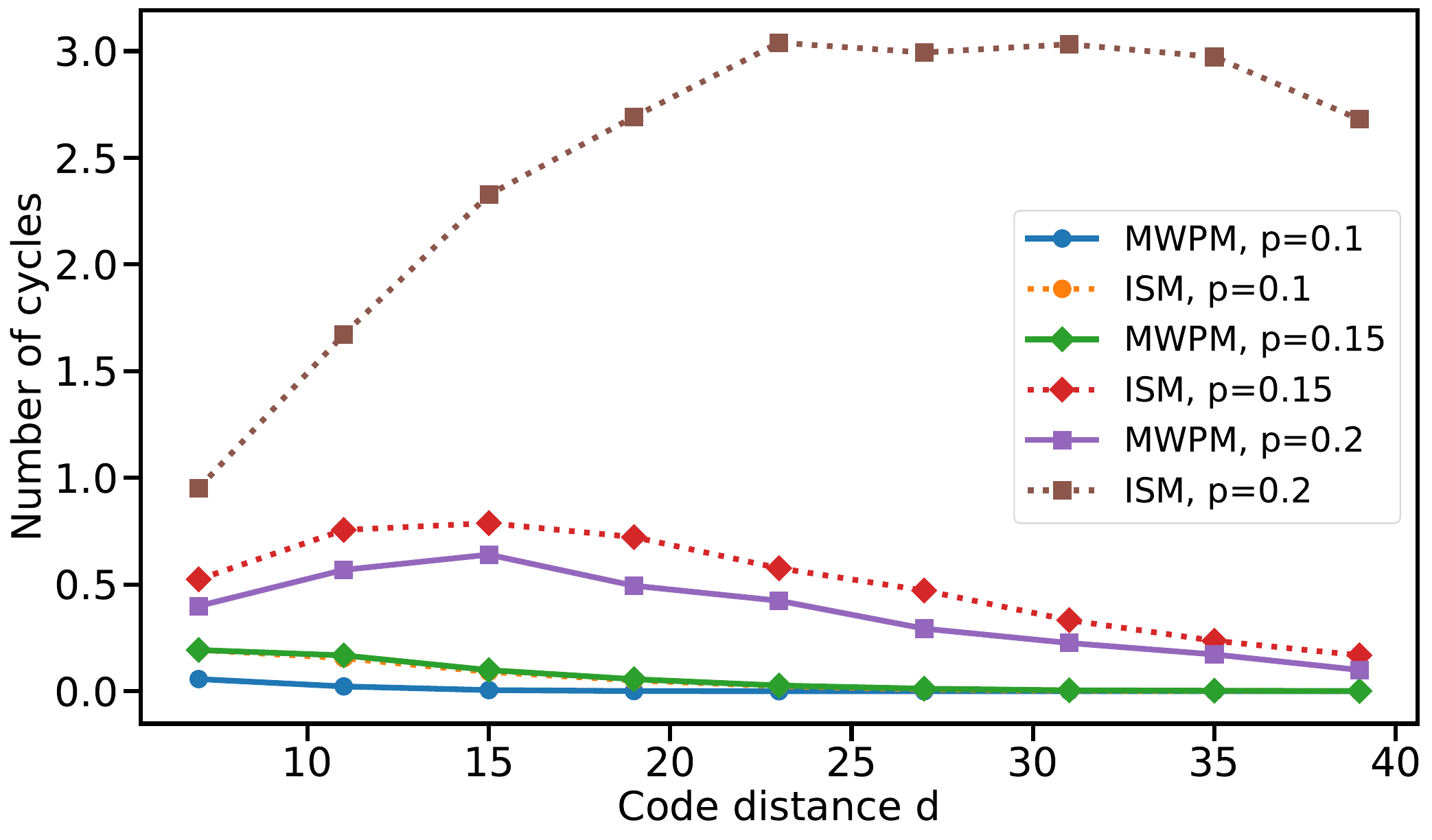}
\caption{\label{fig:number-of-cycles} Average number of post-processing cycles needed for convergence, i.e., the average number of times the while-loop in algorithm~\ref{alg:pp-ideal} is executed before the correction is shorter than half of the qubits long along all 1-lines. For a fixed physical error probability $p$ below the threshold, the number of cycles monotonically decreases after the initial growth with distance $d$. For $d$ large enough, less than one cycle is required on average to achieve convergence.}
\end{figure}

\section{Conclusion} \label{sec:outlook}

In this paper, we have investigated symmetrical and statistical properties of the parity code. We have shown that the former enables a symmetry-based matching decoder, while the latter implies a 50\% error threshold. Combining these properties, we have constructed a two-step decoder that performs near-optimally in the setting with ideal measurements and translates well to a more realistic setting with noisy measurements. The post-processing step itself takes the form of independent repetition codes and, furthermore, allows one to reduce the matching step to a series of repetition codes. Hence, implementation of our two-step decoder is equivalent to solving a series of independent repetition codes and can be executed in parallel. This results in dramatic improvement of runtime scaling compared to an implementation of matching algorithms on a full spacetime graph.

Our approach is reminiscent of the matrix product states~(MPS) decoder of Bravyi, Suchara, and Vargo~\cite{PhysRevA.90.032326}, implemented efficiently in conjunction with surface codes in Refs.~\cite{PhysRevLett.120.050505,PhysRevX.9.041031}. In the first step, the MPS decoding finds any correction operator that returns the state to the codespace. In the second step, the decoder transforms the correction operator into the approximate maximum likelihood one. The accuracy of the decoder can be improved by increasing bond dimension $\chi$, which at the same time results in a rapid increase of complexity as $O(\chi^3)$. The algorithm becomes exact if $\chi$ is exponentially large in $d$. Similarly, the decoding accuracy of our post-processing can be improved by considering $k$-lines with $k>1$ at computational cost growing exponentially with $k$. The algorithm becomes exact if $k=d$, i.e., when the number of repetition codes to be solved is exponentially large in $d$. 

Finally, we note that the post-processing introduced in this paper is not exclusive to the parity code, but is applicable across a variety of classical LDPC codes or to non-degenerate quantum codes, such as the surface code under infinite bias~\cite{PhysRevLett.124.130501}. We leave the study of post-processing algorithms applied to other codes to future work.

\section{Summary of algorithms and code availability} \label{sec:algorithms}

Here we summarize the algorithms used in numerical simulations in this work.
\begin{enumerate}
    \item Algorithm~\ref{alg:pp-ideal} describes post-processing under ideal measurements
    \item Algorithm~\ref{alg:mwpm-2d} describes the matching step for the case of ideal measurements when MWPM is used
    \item Algorithm~\ref{alg:ism-2d} describes the matching step for the case of ideal measurements when ISM is used
    \item Algorithm~\ref{alg:mwpm-3d} describes the matching step for the case of faulty measurements when MWPM is used
    \item Algorithm~\ref{alg:ism-3d} describes the matching step for the case of faulty measurements when ISM is used
    \item Algorithm~\ref{alg:clusters-from-matching} describes the process of finding clusters from matching in both ideal and noisy measurements cases
    \item Algorithm~\ref{alg:correction-2d} describes the process of finding the correction operator from matching under ideal measurements
    \item Algorithm~\ref{alg:correction-3d} describes the process of finding the correction operator from matching under noisy measurements
    \item Algorithm~\ref{alg:distance-3d} describes calculation of a distance metrics in the noisy-measurement case
\end{enumerate}

The code used for simulating the performance of different decoders is open source and, along with the simulation data, available at~\cite{parity-decoder-repo}.
Our implementation is based on Qecsim Python package~\cite{qecsim} for simulating stabilizer QECCs. To accelerate the simulations, we add the support of Pymatching v2~\cite{Higgott_2025}. 

\begin{acknowledgments}
This study was supported by the Austrian Research Promotion Agency (FFG Project No. FO999924030, FFG Basisprogramm) and by NextGenerationEU via FFG and Quantum Austria (FFG Project No.~FO999896208). 
This research was funded in part by the Austrian Science Fund (FWF) under Grant-DOI 10.55776/F71.
This project was funded within the QuantERA II Programme that has received funding from the European Union’s Horizon 2020 research and innovation programme under Grant Agreement No. 101017733. This publication has received funding under Horizon Europe Programme HORIZON-CL4-2022-QUANTUM-02-SGA via the project 101113690 (PASQuanS2.1). For the purpose of open access, the author has applied a CC BY public copyright license to any Author Accepted Manuscript version arising from this submission.
\end{acknowledgments}

\bibliographystyle{apsrev4-1}
\bibliography{reflist}

\clearpage
\onecolumngrid
\appendix

\section{MWPM under ideal measurements} \label{app:mpwm-graph-ideal}

A syndrome in the bulk of the parity code obeys the same dynamics as a syndrome of the $XZ$ surface code under $Y$ errors. Hence, we construct the full MWPM graph analogously to the matching graph of Ref.~\cite{PhysRevLett.124.130501}. For each measured defect of syndrome $S$, we add two nodes to graph $G$, one for each symmetry a given defect belongs to. We also add two nodes for each virtual stabilizer of the code. One of these nodes belongs to a real symmetry 1 to $d$~[shown in Figs.~\ref{fig:parity-code}(b)--(f)], and the other one belongs to the virtual symmetry, which we label symmetry 0~[shown in Figs.~\ref{fig:parity-code}(g)]. These two nodes corresponding to the same virtual stabilizer are connected by a weight-zero edge. Other nodes are connected only if they belong to the same symmetry. The edge between such nodes is weighted by a geometrical distance along the symmetry. The matching algorithm is summarized below.

\begin{algorithm}[H]\label{alg:mwpm-2d}
    \caption{MWPM, ideal measurements}
    \KwIn{Syndrome $S$, error probability $p$}
    \KwOut{Matching $M$}
    
    Initialize graph $G$\;
    \For{each defect in $S$}{
        Add two nodes corresponding to two symmetries of a defect to $G$\;
    }
    \For{each virtual vertex}{
        Add two virtual nodes corresponding to a virtual and a real symmetry to $G$\;
        Connect two nodes of one virtual stabilizer with a zero-weight edge\;
    }    
    \For{each node pair in $G$ belonging to one symmetry}{
        Add edge weighted by $\text{distance}(\text{node pair},p)$\;
    }
    Find matching $M$ from $G$ using MWPM\;
    \Return Matching $M$\;
\end{algorithm}

Matching $M$ is then used to construct the correction operator $C$, as explained in Appendix~\ref{app:clusters-from-matching}.

\clearpage
\section{ISM under ideal measurements} \label{app:repetition-graph-ideal}

For ISM, defects are first matched along one-dimensional symmetries. The matching problem hence reduces to a series of independent repetition codes, with graph $G$ constructed and solved independently for each symmetry. Each of these repetition codes might activate virtual stabilizers at the end of the symmetry line. In the last step, we construct another graph $G$ of nodes corresponding to virtual stabilizers activated during the real symmetry matching. Combined together, 1D matchings along each symmetry constitute 2D matching $M$. The matching algorithm is summarized below.

\begin{algorithm}[H]\label{alg:ism-2d}
    \caption{ISM, ideal measurements}
    \KwIn{Syndrome $S$, error probability $p$}
    \KwOut{Matching $M$}

    Initialize matching $M$\;
    \tcc{Do matching along each real symmetry line and add to list $M$}
    \For{each symmetry $\textrm{Sym}$}{
        Initialize graph $G$\;
        Add nodes corresponding to defects in $S \cap \textit{Sym}$ to $G$\;
        Add one virtual node for each of two virtual stabilizers of \textit{Sym} to $G$\;
        \If{$|S \cap \textit{Sym}| \mod{2} = 1 $}{
        Add one extra virtual node to $G$\;
        }
        \For{each next-neighbor node pair in $G$}{
        Add edge weighted by $\text{distance}(\text{node pair},p)$\;
        }
        Add zero-weighted edges between all virtual nodes\;
        Find matching $m$ from $G$ using MWPM and add $m$ to $M$\; 
    }
    \tcc{Do matching along virtual symmetry line and add to list $M$}
    Initialize graph $G$\;
    \For{each virtual node $\alpha$ in $M$}
    {
        Add $\alpha$ to $G$\;
    }
    \For{each next-neighbor node pair in $G$}{
    Add edge weighted by geometrical distance along the symmetry;
    }
    Find matching $m$ from $G$ using MWPM and add $m$ to $M$\; 
    \Return Matching $M$\;
\end{algorithm}

Matching $M$ is then used to construct the correction operator $C$, as explained in Appendix~\ref{app:clusters-from-matching}.

\clearpage

\section{MWPM vs ISM} \label{app:mwpm-vs-independent}

In Sec.~\ref{sec:matching-step}, we introduced two matching algorithms: MWPM and ISM. Without post-processing, we observe significantly higher logical error rates when the latter one is used during the first step of the decoder. This is due to the fact that ISM only seeks for local minima for each real symmetry, but does not take into account the weight of matching along the virtual symmetry. Unlike MWPM, this does not guarantee minimization of the total perimeter length. Figure~\ref{fig:app-mwpm-vs-repetition-example} shows an example of a syndrome where the two matching algorithms result in different solutions.

\begin{figure*}[h]
\includegraphics[width=0.75\textwidth]{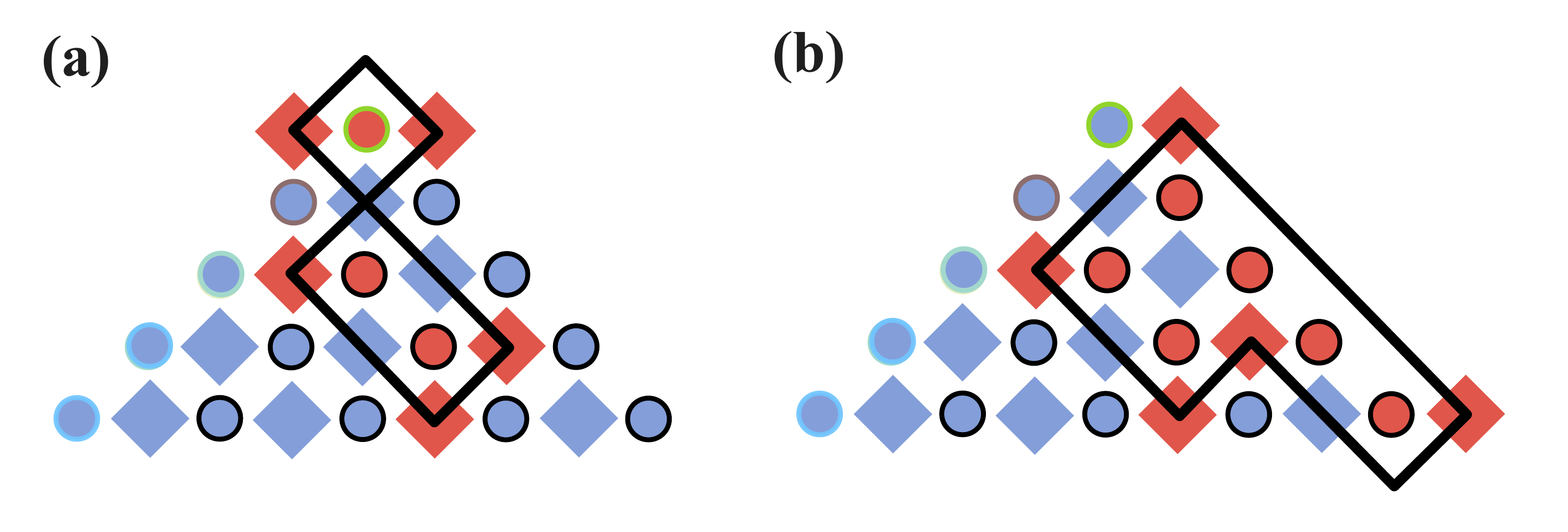}
\caption{\label{fig:app-mwpm-vs-repetition-example} Difference between decoding results of MWPM~(a) and ISM~(b) without post-processing. Since the latter searches for a configuration that minimized only the total length along the real symmetries, it can output a configuration with larger total perimeter than in~(a).}
\end{figure*}

\section{Importance of post-processing} \label{app:need-of-post-processing}

In Fig.~\ref{fig:post-processing}, we provided a toy example, where two possible solutions found by the MWPM algorithm yield identical cluster perimeter, yet different weights of the correction operator. We showed that post-processing allows to lift the ambiguity and find the smallest-weight solution among the two. Figure~\ref{fig:ambiguity-1} shows an example, where the lower-perimeter solution yields a higher-weight correction operator. Again, post-processing allows to find the correct solution in this case. Hence, post-processing can work even in cases where matching alone is guaranteed to fail, providing an even stronger motivation for the need of post-processing.

\begin{figure*}[h]
\includegraphics[width=0.75\textwidth]{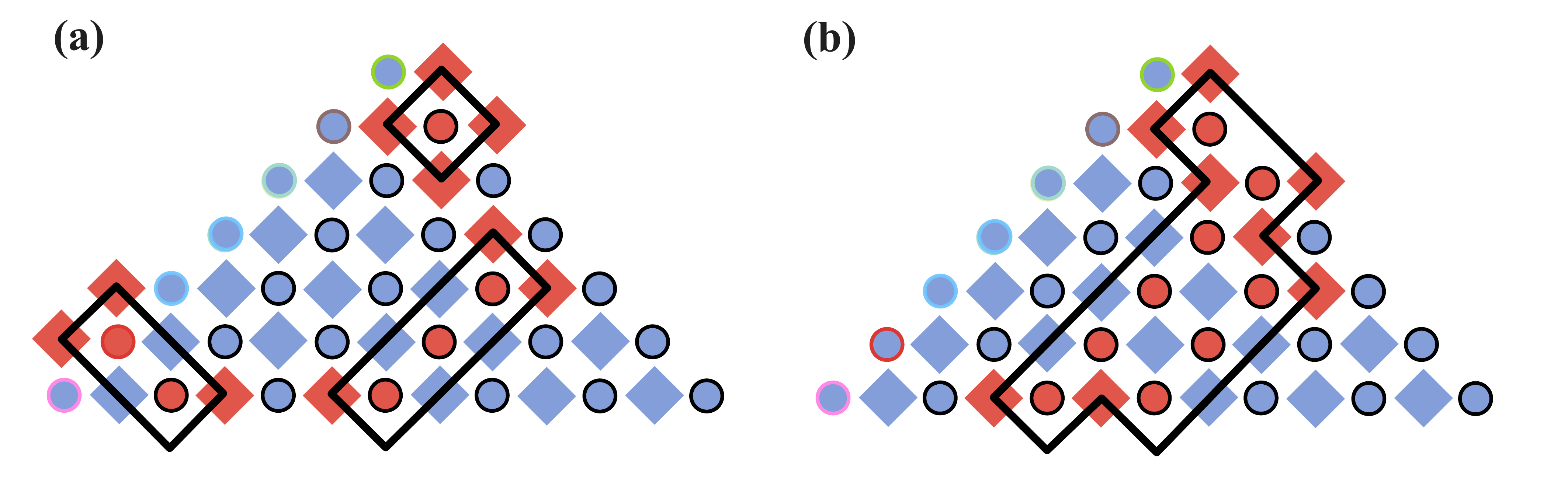}
\caption{\label{fig:ambiguity-1} Two matching solutions corresponding to the same syndrome in distance $d=7$ parity code. Configuration in~(b) has a smaller total cluster perimeter, however, results in a higher-weight correction operator. MWPM alone is hence guaranteed to fail in this case. Applying post-processing transforms correction operator in panel~(b) to the lowest-weight operator of panel~(a).}
\end{figure*}

\clearpage
\section{MWPM under noisy measurements} \label{app:mwpm-graph-3d}

Here, we provide an algorithm for constructing a MWPM graph in $(2+1)$ dimensions. Construction of the algorithm is similar to the case of ideal measurements. Defects of syndrome $S$ are labelled by time step and spatial coordinate and connected by edge weighed according to a distance metrics defined in Appendix~\ref{app:distance-3d}.

\begin{algorithm}[H]\label{alg:mwpm-3d}
    \caption{MWPM, faulty measurements}
    \KwIn{Spacetime syndrome $S$, probability of physical qubit error $p$, probability of measurement error $q$}
    \KwOut{Matching  $M$}
    Initialize graph $G$\;
    \For{each defect in $S$}{
        Add two nodes corresponding to two symmetries of the defect to $G$\;
    }
    \For{each virtual stabilizer at each time step}{
        Add virtual nodes corresponding to a virtual and a real symmetries of the stabilizer to $G$\;
        Connect two nodes of each virtual stabilizer with a zero-weight edge\;
    }    
    \For{each node pair in $G$ belonging to one spacetime symmetry}{
        Add edge weighted by $\text{distance}(\text{node pair},p,q)$\;
    }   
    Find matching $M$ from $G$ using MWPM\;
    \Return Matching $M$\;
\end{algorithm}

Matching $M$ is then used to construct spacetime clusters, as explained in Appendix~\ref{app:clusters-from-matching}.

\clearpage
\section{ISM under noisy measurements} \label{app:repetition-graph-3d}

Here, we provide an algorithm for constructing an ISM graph in $(2+1)$ dimensions. Construction of the algorithm is similar to the case of ideal measurements. For each spacetime symmetry, we construct the graph $G$ by adding measured defects belonging to the symmetry. We also add one node for each virtual stabilizer of the symmetry at each timestep. In addition, if at a given timestep an odd number of defects are measured, we add an extra virtual node well outside of the code boundaries. All virtual nodes are connected by weight-zero edges. All other nodes are pairwise connected by edges with weights defined in Appendix~\ref{app:distance-3d}. After matching graphs have been constructed and solved for each symmetry, we repeat the same procedure for the virtual symmetry, with nodes determined by the ones activated during the real symmetries matching. The matching algorithm is summarized below.

\begin{algorithm}[H]\label{alg:ism-3d}
    \caption{ISM, faulty measurements}
    \KwIn{Spacetime syndrome $S$, probability of physical qubit error $p$, probability of measurement error $q$}
    \KwOut{Matching  $M$}
    Initialize matching $M$\;
    \tcc{Do matching along each real spacetime symmetry and add to list $M$}
    \For{each spacetime symmetry $\textrm{Sym}$}{
        Initialize graph $G$\;
        Add one node for each spacetime defect in $S \cap \textit{Sym}$ to $G$\;
        Add one node corresponding to each virtual spacetime stabilizer of \textit{Sym} to $G$\;
        \For{each timestep}{
            \If{$|S \cap \textit{Sym}| \mod{2} = 1 $ at timestep}
            {Add one extra virtual node at this timestep\;
            Add zero-weight edges between extra virtual nodes and virtual nodes of $\textit{Sym}$ at timestep to $G$\;
            }
        }
        \For{each node pair in $G$}{
        Add edge weighted by $\text{distance}(\text{node pair},p)$\;
        }
        Find matching $m$ from $G$ using MWPM and add $m$ to $M$\; 
    }
    \tcc{Do matching along virtual spacetime symmetry and add to list $M$}
    Initialize graph $G$\;
    \For{each virtual node $\alpha$ in $M$}
    {
        Add $\alpha$ to $G$\;
    }
    \For{each node pair in $G$}{
    Add edge weighted by $\text{distance}(\text{node pair},p)$\;
    }
    Find matching $m$ from $G$ using MWPM and add $m$ to $M$\; 
    \Return Matching $M$\;
\end{algorithm}

Matching $M$ is then used to construct spacetime clusters, as explained in Appendix~\ref{app:clusters-from-matching}.

\clearpage
\section{Additional data for Sec.~\ref{sec:finite-size-ft}} \label{app:data-for-ft-matching}

\begin{figure}[h]
\includegraphics[width=0.95\textwidth]{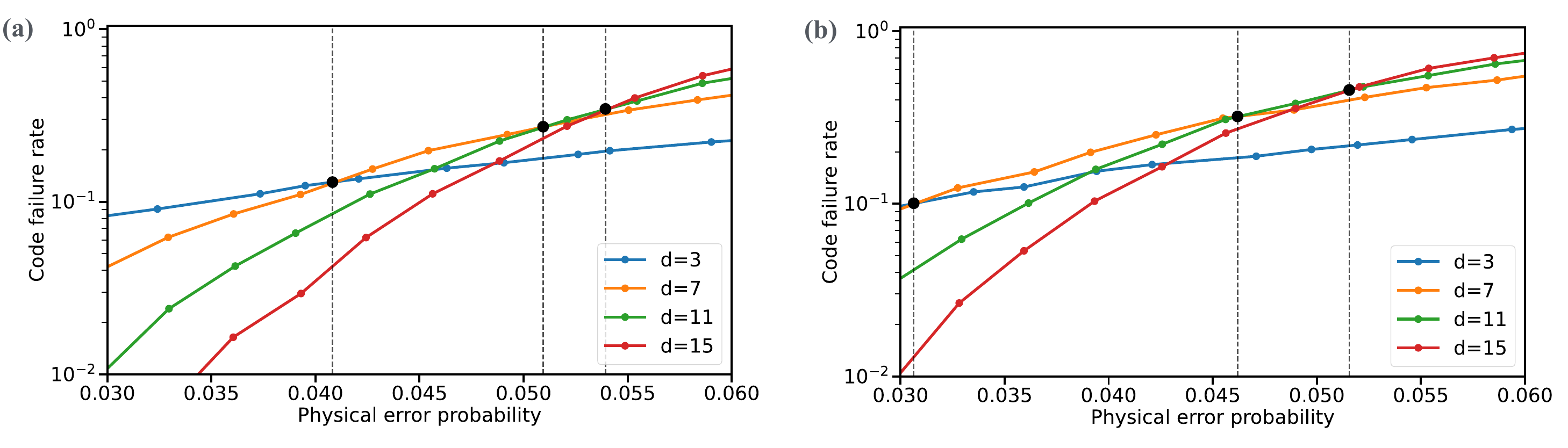}
\caption{\label{fig:si-comparison-1} Code failure rate versus physical error rate for different distances $d$. The measurement error rate is assumed identical to the data qubit rate. The data is shown for ISM (a)~with and (b)~without post-processing. Data for MWPM is shown in Fig.~\ref{fig:LER-FT} of the main text.}
\end{figure}

\clearpage
\section{Building clusters from matching} \label{app:clusters-from-matching}

Algorithm~\ref{alg:clusters-from-matching} describes the construction of clusters in both ideal and noisy measurement settings.

\begin{algorithm}[H]\label{alg:clusters-from-matching}
    \caption{Clusters from matching}
    \KwIn{Matching $M$}
    \KwOut{Set of clusters}
    Remove matches between virtual nodes with the same location from $M$\;
    \While{matching $M$ is not empty}{
        Initialize an empty cluster\;
        Let $(\alpha, \beta)$ be the first pair in $M$\;
        \Repeat{no more selections are possible
        }{
Select and remove vertical node pair $(\alpha, \beta)$ from matching\;
            Add defect corresponding to $\alpha$ to the cluster\;
            Select and remove horizontal node pair $(\beta, \gamma)$ from matching\;
            Add defect corresponding to $\beta$ to the cluster\;
            Let $\alpha \gets \gamma$\;        
        }
        Add cluster to the set of clusters\;
    }
    \Return Set of clusters\;
\end{algorithm}

The set of clusters returned by Algorithm~\ref{alg:clusters-from-matching} is used to construct the correction operator as explained in Sec.~\ref{app:correction-from-clusters-ideal} and Sec.~\ref{app:correction-from-clusters-ft} for ideal and noisy measurements, respectively.

\section{Finding correction from clusters, ideal measurements} \label{app:correction-from-clusters-ideal}

A correction operator is a tensor product of Pauli-$X$ operators applied to qubits enclosed by clusters. Note that multiple clusters can overlap. The correction is applied to qubit $i$ only if it belongs to an odd number of clusters; hence, modulo 2 addition in the algorithm.

\begin{algorithm}[H]\label{alg:correction-2d}
    \caption{Correction from clusters, ideal measurements}
    \KwIn{Set of clusters}
    \KwOut{Correction $C$, set of qubit indices to be corrected}
    Initialize an empty set $C$\;
    \For{each cluster in clusters}{
        Find interior of the cluster $c$ and add to correction modulo 2, $C \leftarrow C \oplus c$ ;
    }
    \If{post-processing}{
    Apply post-processing of Algorithm~\ref{alg:pp-ideal} to $C$
    }
    \Return Correction $C$\;
\end{algorithm}

\section{Finding correction from clusters, noisy measurements} \label{app:correction-from-clusters-ft}

In the noisy-measurement case, spacetime clusters contain information about both qubit and measurement errors. Without post-processing, we project clusters on one time slice. Note that the correction operator is modulo 2 addition of clusters overlap in both spatial~(as in Algorithm~\ref{alg:correction-2d}) and temporal dimensions. With post-processing, we apply Algorithm~\ref{alg:pp-ideal} within each time slice, hence reducing the total weight of the operator compatible with spacetime syndrome.

\begin{algorithm}[H]\label{alg:correction-3d}
    \caption{Correction from clusters, noisy measurements}
    \KwIn{Set of clusters}
    \KwOut{Correction $C$, set of qubits}
    \For{each cluster in clusters}{
        Remove clusters containing only measurement errors\;
    }
    
    \For{each cluster in clusters}{
            Find qubits labels $\{i\}$ enclosed by clusters within each timestep $t$ \;
            Add nodes $\{t,i\}$ to a spacetime correction set $C_{3D}$\;
    }
    \If{post-processing}{        
        \For{each timestep $t$}{
            Apply post-processing of Algorithm~\ref{alg:pp-ideal} to qubits within timestep $t$ of $C_{3D}$\;
        }
    }
    $C \leftarrow$ use $C_{3D}$ to calculate cumulative correction from Eq.~\eqref{eq:C-projected}\;
    \Return Correction $C$\;
\end{algorithm}

\section{Distance metrics, noisy measurements} \label{app:distance-3d}

An algorithm for calculating the distance metrics was introduced in Ref.~\cite{PhysRevLett.124.130501} and is summarized below.

\begin{algorithm}[H]\label{alg:distance-3d}
    \caption{Distance metrics calculation, noisy measurements}
    \KwIn{Nodes indices, probability of physical qubit error $p$, probability of measurement error $q$}
    \KwOut{Weight $W$}
    $\delta_t \gets$ minimum number of time steps between $\alpha$ and $\beta$\;
    $\delta_p \gets$ minimum number of parallel steps between $\alpha$ and $\beta$\;
    $\mu_p \gets \log{(p/(1-p))}$\;
    $\mu_t \gets \log{(q/(1-q))}$\;
    $W \gets \delta_p \mu_p + \delta_t \mu_t$\; 
    \Return Weight $W$\;
\end{algorithm}

\end{document}